\documentclass[a4paper]{article} 

\usepackage{balance} 
\usepackage{amsmath}
\usepackage{amsthm}
\usepackage{amsfonts}
\usepackage{amssymb}
\usepackage{graphicx}
\usepackage{epsfig} 
\usepackage{color}
\usepackage{epstopdf}
\usepackage{mathtools}
\usepackage{algorithm}
\usepackage{graphicx}
\usepackage[noend]{algpseudocode}
\usepackage{amsthm}
\usepackage{tabularx, caption}
\usepackage{amsfonts}
\usepackage{subcaption}
\usepackage{gensymb}
\usepackage{soul}
\usepackage{tikz}
\usetikzlibrary {positioning}
\usepackage{pgfplots}
\usepackage{titlesec}
\usepackage{tikz}
\usepackage{tkz-euclide}
\usepackage{multicol}
\usepackage[skins]{tcolorbox}
\usepackage{hyperref}

\usepackage{geometry}
\usepackage{sectsty}

\sectionfont{\fontsize{11}{14}\selectfont}

\DeclareMathOperator*{\argmax}{arg\,max}
\newtheorem{theorem}{Theorem}[section] 
\newtheorem{corollary}{Corollary}
\newtheorem{assumption}{Assumption}
\newtheorem{lemma}[theorem]{Lemma}

\theoremstyle{definition}
\newtheorem{definition}{Definition}[section]

\newenvironment{Algorithm}
  {\par\medskip\noindent\minipage{\linewidth}}
  {\endminipage\par\medskip}
  
\newenvironment{Table}
  {\par\medskip\noindent\minipage{\linewidth}}
  {\endminipage\par\medskip}
 
\newenvironment{Figure}
  {\par\medskip\noindent\minipage{\linewidth}}
  {\endminipage\par\medskip}

\providecommand{\keywords}[1]
{
  \small	
  \textbf{\textit{Keywords---}} #1
}

\title{Abstract, keywords and references template}
\author{Saar Cohen$^{1}$, Noa Agmon$^{2}$  \\
        \small Department of Computer Science \\
        \small Bar Ilan University, Israel \\
        \small saar30@gmail.com$^{1}$, agmon@cs.biu.ac.il$^{2}$\\
}

\title{Converging to a Desired Orientation in a Flock of Agents}
\date{\empty}

\begin{document}
\maketitle

\begin{abstract}
This work concentrates on different aspects of the \textit{consensus problem}, when applying it to a swarm of flocking agents. We examine the possible influence an external agent, referred to as {\em influencing agent} has on the flock. We prove that even a single influencing agent with a \textit{Face Desired Orientation behaviour} that is injected into the flock is sufficient for guaranteeing desired consensus of the flock of agents which follow a Vicsek-inspired Model. We further show that in some cases this can be guaranteed also in dynamic environments. 
\end{abstract} \hspace{10pt}

\keywords{Coordination, Consensus Algorithms, Flocking, Swarm Robotics, Multi-Agent Systems, Cooperative Control, Graph Laplacians, Networked Control Systems, Graph Theory, Lyapunov Function, Switched Systems, Computational Geometry}


\newcommand{\BibTeX}{\rm B\kern-.05em{\sc i\kern-.025em b}\kern-.08em\TeX}

\pgfplotsset{compat=1.16}

\maketitle 

\begin{multicols}{2}

\section{Introduction}
\label{sec:intro}
Reaching a consensus within a group of simple robots that have limited sensing, use no explicit communication and have limited computational capabilities is one of the core problems in swarm robotics. The goal is to find a protocol for the robots such that eventually they will all agree on a certain value. In this work we consider the problem of converging a flock of agents to a single orientation, that is, reaching a consensus on the orientation. As opposed to previous work \cite{ren2005consensus,olfati2004consensus,olfati2007consensus}, we examine the problem of forcing the flock to reach a {\em desired} orientation by inserting one or \textit{more} agents, referred to as {\em influencing agents}, into the flock.

We assume that the agents are initially placed in one connected component, facing random directions. We prove that influencing agents with a \textit{Face Desired Orientation behaviour} are sufficient for guaranteeing consensus, while employing them into a swarm of flocking agents that follow the Vicsek Model \cite{vicsek1995novel}.

Genter and Stone \cite{genter2016ad,genter2016adding} introduced this kind of baseline behaviour, according to which influencing agents face a predetermined desired orientation at all times. This behaviour was actually inspired by Jadbabaie et al. \cite{jadbabaie2003coordination}, who considered a modified version of Vicsek’s model which utilizes a \textit{single} leader with a fixed orientation and several followers. Hence, the \textit{Face Desired Orientation behaviour} is the one described in their work, except that we propose several significant extensions regarding this behaviour:
\begin{enumerate}
    \item We consider \textit{both} the fixed topology case and the switching topology case, instead of considering just a single one of them.
    \item Opposed to their work, we consider \textit{multiple} Vicsek-inspired Models, instead of just a \textit{single} one.
    \item Furthermore, Jadbabaie et al. \cite{jadbabaie2003coordination} utilized an \textit{undirected} graph model with a \textit{single} leader, but we consider a \textit{directed} one with \textit{multiple} influencing agents.
    \item Whereas Genter and Stone \cite{genter2016ad,genter2016adding} solely give \textit{experimental results} regarding this behaviour, we also propose \textit{theoretical proofs} of useful properties while such behaviour is incorporated.
\end{enumerate}

For a \textit{fixed topology}, we also prove that continuous-time update rules converge exponentially \textit{faster} than discrete-time ones, thus yielding the benefits of the continuous-time case. For a \textit{switching topology with a continuous-time update rule}, we propose a much more concrete convergence rate than the one proposed by Olfati-Saber and Murray \cite{olfati2004consensus}.

We also propose a computational geometric algorithm, referred to as the \textit{Intersection Points Placement Method}, which guarantees in polynomial time that a \textit{single} influencing agent can be inserted into the flock in a manner that it will necessarily influence an entire single connected component. This algorithm utilizes the locus of all point in $\mathbb{R}^2$, which lie on the linear line connecting the intersection points of neighborhoods which correspond to a random pair flocking agents.

We have situated our research in the MASON simulator \cite{luke2005mason}, demonstrating the impact of the number of influencing agents on the convergence rate in different scenarios. Finally, we discuss and demonstrate the concept of "lost" flocking agents \cite{genter2017fly} in a switching topology. 
\section{Related Work}
In this section, we provide a review of the state-of-the-art in different domains that are related to our work. The articles mentioned below concentrate on reaching a consensus among swarms, while utilizing various tools for the sake of achieving it.
\subsection{Ad-Hoc Teamwork}
Due to the rise in the use of autonomous agents, both in software and robotic settings, they will increasingly need to band together for cooperative activities with previously unfamiliar teammates. As described by Stone et al. \cite{stone2013teaching}, multiple agents with different knowledge and capabilities find themselves in a situation such that their goals and utilities are perfectly aligned (for instance, everyone’s sole interest is to help find survivors), yet they have had no prior opportunity to coordinate.

Genter and Stone \cite{genter2013ad} consider the problem of leading a team of flocking agents in an ad-hoc teamwork setting. An ad-hoc teamwork setting is one in which a teammate - which they call an \textit{influencing agent} - must determine how to best achieve a team goal given a set of possibly suboptimal teammates. The main question addressed in this paper is: \textit{how should influencing agents behave so as to orient the rest of the flock towards a target heading as quickly as possible?} The main contribution of this paper is the 1-step lookahead algorithm for orienting a flock to travel in a particular direction and its variations.

Genter et al. \cite{genter2015determining}  further examine where these influencing agents should be placed in the flock. Specifically, their research question is: \textit{where should influencing agents be located within a flock to maximize their influence on the flock?} They introduce a graph-based method for deciding where to place the influencing agents and their approach to control the influencing agents. They suggest two cases for controlling: the \textbf{Drop} case (in which the  k influencing agents are at their desired positions at $t = 0$) and the \textbf{Dispatch} case (in which  the influencing agents are initially positioned outside of the flock).

Researches investigate those ad-hoc settings from different angles \cite{genter2013ad,stone2013teaching}. For instance, \cite{stone2013teaching} focuses on settings in which we are designing a new agent that has full information about its environment, that must coordinate with an older, less aware, more reactive agent whose behaviour is known, while using game-theoretic tools.
These researches usually deploy the main idea of letting each robot execute a simple algorithm and plan its motion adaptively based on the observed movement of other agents, so that the agents as a group will achieve the given goal. 

Fu et al. \cite{fu2018influencing} observe on this problem from a different perspective. They focus on lower-density settings where interactions are rarer and flock formation is more difficult. They study how such influencing agent priorities must change in these settings to be successful and propose new influencing agent strategies to adapt to the challenges posed by these settings. Low-density settings are important to study because they capture dynamics in situations where flocking may not occur naturally, but where we might want to instigate flocking behaviour; imagine a herd of buffalo that is currently grazing, or a spooked flock of birds where individual agents fail to coordinate.

Considerable research into ad-hoc teamwork has established a number of theoretical and empirical methods for handling unknown teammates, but these analyses largely focus on simple scenarios that may not be representative of the real problems that agents may encounter in influencing teams. Barret and Stone \cite{barrett2015cooperating} come to address this gap. The main contribution is the introduction of an algorithm (PLASTIC–Policy) for selecting effective actions for cooperating with a variety of teammates in complex domains. This algorithm learns policies of how to cooperate with previous teammates and then reuses this knowledge to cooperate with new teammates. Learning these policies is treated as a \textit{reinforcement learning problem}, where the inputs are features derived from the world state and the actions are soccer moves such as passing or dribbling. Then, PLASTIC–Policy observes the current teammates’ behaviours to select which of these policies will lead to the best performance for the team.

\subsection{Consensus in Groups of Autonomous Agents}
In the traditional consensus problem of distributed computing \cite{fischer1983consensus}, the goal is for all processors in a network to agree on a certain value. The problem becomes challenging as the network becomes more complex, and the consensus protocol should handle sparse networks, asynchronous timing and possible faults in the system \cite{fischer1985impossibility,olfati2004consensus,olfati2007consensus}.

Regarding the consensus problem in swarm robotics, the dynamic nature of the network is brought to its extreme: the network topology changes throughout the execution, and a robot controls not only the decision of how to change its own value, but also its movement. The robot’s choice of action in the physical world has an impact on what its neighbors observe, and thus also its own state in the future \cite{fax2004information,olfati2004consensus,olfati2007consensus}. Modeling those topologies is usually done using graph theory and matrix theory. 

According to Ren and Beard \cite{ren2005consensus}, those topologies are depicted using directed graphs and the minimum requirements to reach consensus are explored in the case of strongly connected graphs. More specifically, it is shown that the agents reach a consensus asymptotically if and only if the associated interaction graph has a spanning tree.

Olfati-Saber and Murray \cite{olfati2004consensus} examine directed networks with fixed or switching topologies. They present two consensus protocols which solve agreement problems in a network of either continuous-time (CT) integrator or discrete-time (DT) model. One protocol depicts either a fixed or a switching topology with zero communication time-delay, while the other depicts a fixed topology with communication time-delay corresponding to every edge. They also present the model using a linear system, induced by the \textit{Laplacian} matrix or the \textit{Perron} matrix \cite{olfati2004consensus}, both induced by the interaction graph of the topology. Their paper also examines the spectral properties of both matrices, while proving conditions under which the flock reaches a consensus.

On the contrary, some researches examine the coordination of flocking agents using the Vicsek Model \cite{vicsek1995novel}. According to this model, each agent's heading is updated using a simple local rule based on the average of its own heading plus the headings of its neighbors.

Jadbabaie et al. \cite{jadbabaie2003coordination}, both leaderless and leader following coordinations are examined. For each case, this article gives us conditions under which the flock reaches a consensus upon its orientation. In the case of leaderless coordination, it is shown that given an infinite sequence of contiguous, nonempty, bounded, time-intervals, with the property that across each such interval the agents are linked together, the flock will reach a consensus. In the case of leader following coordination, a similar condition is proven.

Some researches use the Laplacian matrix's spectral properties for exploring the consensus problem. The second eigenvalue of the Laplacian is widely used in Olfati-Saber and Murray \cite{olfati2004consensus}. This eigenvalue is known as the \textit{algebraic connectivity} of the graph associated with the Laplacian matrix, and aside from depicting the connectivity of the graph, this article gives us the speed in which the flock reaches a consensus.

Moreover, Olfati-Saber and Murray \cite{olfati2004consensus} also use the spectral properties of the Laplacian to provide conditions for solving the average-consensus problem. It considers the case of topologies with either zero communication time-delay or with communication time-delay and shows conditions under which consensus is achieved.

Olfati-Saber et al. \cite{olfati2007consensus} also examine the spectral properties of the Laplacian in order to investigate the average-consensus problem (in which the consensus value equals to the average of the the initial value of the state of each agents). It gives conditions under which a consensus is asymptotically reached for all initial states, using these spectral properties. It considers both the continuous-time (CT) case and the discrete-time (DT) case. In the DT case, the \textit{Perron} matrix is examined, while examining its connection to the Laplacian matrix.

Hong et al. \cite{hong2007lyapunov} consider neighbor-based rules which are adopted to realize local control strategies for these continuous-time autonomous agents described by double integrators. Although the inter-agent connection structures vary over time and related graphs may not be connected, a sufficient condition to make all the agents converge to a common value is given for the problem by a proposed Lyapunov-based approach and related space decomposition technique.

The Lyapunov-based approach also provides a way of showing the speed in which a dynamics vanishes \cite{olfati2007consensus,olfati2004consensus}. A disagreement dynamics is shown, which measures how far the agents are from a consensus. This article uses the Lyapunov-based approach to show the speed in which the agents reach a consensus asymptotically. 

\section{The Flocking Model}
\label{sec:flocking-model}
The flock comprises of $k$ flocking agents, trying to converge to the same orientation $\theta^*$, and $m$ influencing agents. On the one hand, the flocking agents $A^F : = \{a_0,...,a_{k-1}\}$ are agents which we cannot directly control and we would like to consider the implications of various definitions for their global orientation separately. On the other hand, we can control the behaviour of the influencing agents $A^D := \{a_k,...,a_{n-1}\}$, where $n = m + k$. We assume that the influencing agents are capable of communication, they have full knowledge of the world and a high computability.

In contrast to the flocking agents, each influencing agent's position is updated according to their assumed behaviour. That is, it is not necessarily updated after its orientation is updated in accordance to a specific update rule.

At each time step $t$, the number of influencing agents inside $a_i$'s neighborhood is denoted by $m_i (t)$ and the number of flocking agents inside the neighborhood is denoted by $k_i (t)$, where $n_i (t) = m_i (t) + k_i (t) \leq n$.

\textit{Lemma 2} from \cite{genter2013ad} will be used throughout the entire research. This lemma is phrased as follows:
\begin{corollary}
	\label{Same Orientaion}
	When $m_i(t)$ influencing agents work together to influence $k_i(t)$ flocking agents to align the team to some orientation, it suffices to consider only algorithms that choose at each time step just one orientation for all the  agents to adopt.
\end{corollary}

Flocking agents update their orientation based on the orientations of the other agents in their neighborhood, defined by the \textit{visibility radius}. Let $N_i (t)$ be the set of $n_i (t) \leq n$ agents (including agent $a_i$) at time $t$ which are located within a \textit{visibility radius} $R$ of agent $a_i$. This means that each agent's \textit{visibility cone} equals to $2\pi$. We shall make the following definitions:
\begin{definition}
    \label{def:flocking-neighbors-graph}
	(\textit{Flocking Neighbors Graph})
	The relationships between neighbors which exist in time step $t$ and influence the form of the update equations can be described by a simple, \textbf{directed} graph (digraph) $\mathcal{G}(t) = (\mathcal{V},\mathcal{E}(t))$, where: 
	\begin{itemize}
	\item $\mathcal{V} = \{0,...,k-1\}$ - The vertex set represents the flocking agents.
	\item $\mathcal{E}(t) = \{(i,j) \in \mathcal{V} \times \mathcal{V} | a_i \neq a_j \in N_i(t)\}$
	\end{itemize}
	This graph will be named the \textit{Flocking Neighbors Graph}. We denote its number of connected components at time $0$ by $\eta$.
\end{definition}
\begin{definition}
\label{def:strongly-connected-digraphs}
(\textit{Strongly Connected Digraphs})
A digraph $\mathcal{G}=(\mathcal{V},\mathcal{E})$ is \textit{strongly connected} if there is a directed path connecting any two arbitrary nodes $s,t \in \mathcal{V}$ of the graph.
\end{definition}
\begin{definition}
\label{def:balanced-digraphs}
(\textit{Balanced Graphs} \cite{olfati2004consensus})
A graph $\mathcal{G}=(\mathcal{V},\mathcal{E})$ (which is either directed or undirected) with an adjacency matrix $A$ is \textit{balanced} if $\sum_{j \neq i} A[i,j] = \sum_{j \neq i} A[j,i]$ for all $i \in \mathcal{V}$.
\end{definition}
\begin{definition}
\label{def:fixed-topology}
(\textit{Fixed Topology})
If there exists some digraph $\mathcal{G}=(\mathcal{V},\mathcal{E})$ for which $\mathcal{G}(t) = \mathcal{G}$ for any time step $t$, we say that this is a \textit{fixed topology}.
\end{definition}
\begin{definition}
\label{def:switching-topolgy}
(\textit{Switching Topology})
A \textit{switching topology} can be modeled using a dynamic graph $\mathcal{G}(t) = \mathcal{G}_{\sigma(t)}$ parameterized with a switching signal $\sigma(t):\mathbb{N} \rightarrow \mathcal{Q}_k$, where $\mathcal{Q}_k$ denotes a suitably defined set, indexing the class of all simple graphs defined on $k$ vertices.
\end{definition}

We denote the \textit{adjacency matrix} and the \textit{degree matrix} of the flocking neighbors graph at time step $t$ by $A(t),D(t)$ (respectively). We also denote by $L(t) = D(t) - A(t)$ the \textit{graph Laplacian} of the flocking neighbors graph at time step $t$, whose elements are defined as follows:
$$L(t)[i,j] = 
\begin{cases}
-1, \quad a_j \in N_i(t) \\ 
|N_i(t)|, \quad j=i
\end{cases}$$

We denote by $P(t) = I - \varepsilon L(t)$ the \textit{Perron matrix} of the flocking neighbors graph at time step $t$, where $I$ is the identity matrix and $\varepsilon > 0$ is the step-size.

We denote the global orientation of a flocking agent $a_i$, $0 \leq i \leq k-1$, at time step $t+1$, by $\theta_i (t+1)$. Each agent $a_i$ moves with velocity $v_i$. At each time step $t$, each agent $a_i$ has a position $p_i(t) = (x_i (t),y_i (t))$ in the environment and an orientation $\theta_i (t)$. According to \cite{genter2015determining}, each agent's position $p_i (t)$ at time $t$ is updated after its orientation is updated, such that:
$$\begin{cases}
x_i(t) = x_i(t-1) + v_i cos(\theta_i(t)) \\
y_i(t) = y_i(t-1) - v_i sin(\theta_i(t))
\end{cases}$$

The influencing agents join the flock in order to influence its members to behave in a particular way. We consider the \textbf{Drop} case \cite{genter2015determining} for the initial placements $\{p_k (0),...,p_{n-1} (0)\}$ of the influencing agents, in which the influencing agents are at their desired location at time $t=0$. 

Alternatively, we could define a neighboring graph which also takes the influencing agents into consideration. This graph's definition is as follows:
\begin{definition}
    \label{def:influencing-neighbors-graph}
	(\textit{influencing Neighbors Graph})
	The relationships between neighbors (both flocking and influencing) which exist in time step $t$ and influence the form of the update equations can be described by a simple, directed graph (digraph) $\tilde{\mathcal{G}}(t) = (\tilde{\mathcal{V}},\tilde{\mathcal{E}}(t))$, where: 
	\begin{itemize}
	\item $\tilde{\mathcal{V}} = \{0,...,n-1\}$ - The vertex set representing both the flocking agents and the influencing agents.
	\item $\tilde{\mathcal{E}}(t) = \{(i,j) \in \mathcal{V} \times \mathcal{V} | a_i \neq a_j \in N_i(t)\}$
	\end{itemize}
	This graph will be named the \textit{influencing Neighbors Graph}. Clearly, for all $t \geq 0$, we denote the adjacency matrix, the degree matrix, the Laplacian matrix and the Perron matrix of $\tilde{\mathcal{G}}(t)$ by $\tilde{A}(t),\tilde{D}(t),\tilde{L}(t), \tilde{P}(t)$ (respectively).

	\underline{\textbf{Note:}} Definitions \ref{def:fixed-topology} and \ref{def:switching-topolgy} also hold for an influencing neighbors graph.
\end{definition}

Throughout the entire research, we will be considering the connected components of either the flocking neighbors graph or the influencing neighbors graph and the subgraph which they induce upon them. For this sake, we make the following definition:
\begin{definition}
    \label{def:induced-subgraph}
    (\textit{Induced Subgraph})
    Let $\mathcal{G} = (\mathcal{V},\mathcal{E})$ be some simple graph. For some subset $\mathcal{U} \subseteq \mathcal{V}$ of vertices, $\mathcal{G}(\mathcal{U}) = (\mathcal{U},\mathcal{E}(\mathcal{U}))$ with $\mathcal{E}(\mathcal{U}) = \{(v,u) \in \mathcal{U} \times \mathcal{U} | (v,u) \in \mathcal{E}\}$ is called the \textit{subgraph} of $\mathcal{G}$ \textit{induced} by $\mathcal{U}$.
\end{definition}

It should be noted that the update rule for each agent's position is performed in accordance to the update rule defining the global orientation of each agent. For each \textit{flocking} agent, it is either the discrete-time update rule presented in Subsection \ref{sec:dt-update-rule} or the continuous-time update rule presented in Subsection \ref{sec:ct-update-rule}). These update rules are based on the Vicsek Model \cite{vicsek1995novel}. 

Agent $a_j$'s position in the environment at time step $t$ is located at angle $\beta_{ij}(t) \leq \pi$ with respect to $a_i$'s position.

\subsubsection{Discrete-Time Update Rule}
\label{sec:dt-update-rule}
The global orientation of a flocking agent $a_i$, $0 \leq i \leq k-1$, at time step $t+1$, $\theta_i (t+1)$, can be defined in two ways in the discrete-time (DT) case. According to Olfati-Saber et al. \cite{olfati2007consensus}, if we denote $\theta(t) = [\theta_i(t)]$ ($0 \leq i \leq k-1$), it can be defined as follows:
\begin{equation}
    \label{eq:orientation-perron}
    \theta (t+1) = P(t)\theta (t)
\end{equation}

According to \cite{genter2015determining}, it can be set to be the average orientation of all agents in $N_i (t)$ (including itself) at time step $t$. This definition is actually a special case of the previous one, in the manner that it is a variation of the case $\varepsilon = 1$:	
\begin{equation}
    \label{eq:orientation-calcdiff}
    \theta_i (t+1) = \theta_i (t) + \frac{1}{n_i (t)} \sum_{a_j \in N_i (t)}calcDiff(\theta_j (t),\theta_i (t))
\end{equation}

\begin{Algorithm}
    \label{alg:calc-diff}
    \captionof{algorithm}{$calcDiff(\theta_i (t) , \theta_j (t))$  \cite{genter2015determining}}
    \centering
  \fbox{\begin{minipage}{\columnwidth}
	\begin{algorithmic}[1]
		\If{$((\theta_i (t) - \theta_j (t)) \geq - \pi \wedge (\theta_i (t) - \theta_j (t)) \leq \pi))$} 
		
		\State return $\theta_i (t) - \theta_j (t) $
		
		\Else \If{$\theta_i (t) - \theta_j (t) < - \pi$}
		
		\State return $2\pi + (\theta_i (t) - \theta_j (t))$
		
		\Else
		
		\State return $(\theta_i (t) - \theta_j (t)) - 2\pi$
		
		\EndIf
		
		\EndIf
		
	\end{algorithmic}
  \end{minipage}}
\end{Algorithm}

Instead of taking the mathematical average orientation of all agents, we use Equation 1 because of the special cases handled by Algorithm 1 (from \cite{genter2015determining}). As a convention, we assume $\theta_i (t)$ is within $[0,2\pi)$.  
\subsubsection{Continuous-Time Update Rule}
\label{sec:ct-update-rule}
The global orientation of an agent $a_i$, $0 \leq i \leq n-1$, at time step $t$ can be defined in two ways in the continuous-time (CT) case. According to \cite{olfati2004consensus}, it can be defined using the following linear system:
\begin{equation}
    \label{eq:global-orient-linear}
    \dot{\theta_i}(t) = \sum_{a_j \in N_i(t)} \tilde{A}(t)[i,j] (\theta_j(t) - \theta_i(t))
\end{equation}

The following lemma gives us an equivalent definition in matrix terms:

\begin{lemma}
\label{lemma:global-orient-linear-matrix}
    We denote the adjacency matrix, the degree matrix and the Laplcian matrix of the influencing neighbors graph at time step $t$ by $\tilde{A}(t),\tilde{D}(t),\tilde{L}(t)$ (respectively). Denoting $\tilde{\theta}(t) = [\theta_i(t)]$ ($0 \leq i \leq n-1$), the following heading system is equivalent to Equation \ref{eq:global-orient-linear}: $$\dot{\tilde{\theta}}(t) = - \tilde{L}(t) \tilde{\theta}(t)$$
\end{lemma}
\begin{proof}
    The elements of the degree matrix are defined as follows: 
    $$\tilde{D}(t)[i,j] = 
    \begin{cases}
    \sum_{a_j \in N_i(t)}\tilde{A}[t][i,j], \quad i = j\\ 
    0, \quad j \neq i
    \end{cases}$$
    Since $\tilde{L}(t) = \tilde{D}(t) - \tilde{A}(t)$ by definition, then its elements are defined as follows:
    $$\tilde{D}(t)[i,j] = 
    \begin{cases}
    \sum_{a_j \in N_i(t)}\tilde{A}[t][i,j] - \tilde{A}(t)[i,i], \quad i = j\\ 
    - \tilde{A}(t)[i,j], \quad j \neq i
    \end{cases}$$
    Therefore:
    $$(- \tilde{L}(t) \tilde{\theta}(t))[i] = \bigg[\tilde{A}(t)[i,i] - \sum_{a_j \in N_i(t)}\tilde{A}[t][i,j]\bigg]\theta_i(t) +$$
    $$+ \sum_{a_i \neq a_j \in N_i(t)}\tilde{A}(t)[i,j]\theta_j(t) =$$ $$=\sum_{a_j \in N_i(t)} \tilde{A}(t)[i,j] (\theta_j(t) - \theta_i(t)) = \dot{\theta_i}(t)$$
\end{proof}

Fax and Murray \cite{fax2004information} introduced the following version of a Laplacian-based system in a flock with 0–1 weights: 
\begin{equation}
    \label{eq:global-orient-zero-one}
    \dot{\theta_i}(t) = \frac{1}{|N_i(t)|}\sum_{a_j \in N_i(t)} (\theta_j(t) - \theta_i(t))
\end{equation}

The following lemma gives us an equivalent definition in matrix terms:
\begin{lemma}
\label{lemma:global-orient-zero-one-matrix}
    We denote the adjacency matrix, the degree matrix and the Laplcian matrix of the influencing neighbors graph at time step $t$ by $\tilde{A}(t),\tilde{D}(t),\tilde{L}(t)$ (respectively). Denoting $\tilde{\theta}(t) = [\theta_i(t)]$ ($0 \leq i \leq n-1$), the following heading system is equivalent to Equation \ref{eq:global-orient-zero-one}: $$\dot{\tilde{\theta}}(t) = (\tilde{D}(t) ^ {-1} \tilde{A}(t) - I) \tilde{\theta}(t)$$
\end{lemma}
\begin{proof}
    This is a special case of an influencing neighbors graph $\mathcal{G}^*(t)$, for which $I$ is the degree matrix and $\tilde{D}(t) ^ {-1} \tilde{A}(t)$ is the adjacency matrix. Thus, based on Lemma \ref{lemma:global-orient-linear-matrix}, we have that:
    $$\dot{\tilde{\theta}}(t) = - (I - \tilde{D}(t) ^ {-1} \tilde{A}(t)) \tilde{\theta}(t) = (\tilde{D}(t) ^ {-1} \tilde{A}(t) - I) \tilde{\theta}(t)$$
\end{proof}

\section{Influencing Connected Components}
\label{sec:influencing-connected-components}
In this section, we deal with reaching a consensus among agents, which correspond to a single connected component in either the flocking neighbors graph or the influencing neighbors graph. Each kind of topology (either fixed or switching) affects the influence an influencing agent has upon some connected component. Furthermore, in each case, we shall take into consideration both the discrete-time update rules defined in Subsection \ref{sec:dt-update-rule} (Equations \ref{eq:orientation-perron} and \ref{eq:orientation-calcdiff}) and the continuous-time update rules defined in Subsection \ref{sec:ct-update-rule} (Equations \ref{eq:global-orient-linear} and \ref{eq:global-orient-zero-one}). Both kinds of update rules affect the speed at which the flocking agents converge to their desired orientation.

First, we show a simplification of the update rule defined by Equation \ref{eq:orientation-calcdiff} called the \textbf{normalized Perron matrix} and show some properties regarding it (Subsection \ref{sec:normalized-Perron-matrix}). Then, we consider the influence of the influencing agents, both in the fixed topology case (Subsection \ref{sec:fixed-topology}) and in the switching topology case (Subsection \ref{sec:switching-topology}). More accurately, in each case, we give theoretical proofs regarding that influencing agents with a \textit{Face Desired Orientation behaviour} are sufficient for guaranteeing consensus, while employing them into a flocking model that is based on the Vicsek Model \cite{vicsek1995novel}.

\subsection{Normalized Perron Matrix}
\label{sec:normalized-Perron-matrix}
When considering Equation \ref{eq:orientation-calcdiff}, we would like to also consider taking the mathematical average orientation of all agents, which is a simplification of the update rule defined by Equation \ref{eq:orientation-calcdiff}. 

The following lemma gives us a property of Algorithm \ref{alg:calc-diff}. This lemma claims that the difference between the orientations of two agents is always within $[-\pi,\pi]$ since the difference of $\pi + \varepsilon$ is equivalent to a difference of $\varepsilon - \pi$ in the opposite direction
\begin{lemma}
	\label{lemma:difference-restricted}
	For each pair of agents $a_i,a_j$, the following holds for all time step $t \geq 0$: $calcDiff(\theta_i(t),\theta_j(t)) \in [-\pi,\pi]$.
\end{lemma}
\begin{proof}
	Let $t \geq 0$ be some time step. Let $a_i,a_j$ be two flocking agents. The proof is divided into three parts:
	\begin{enumerate}
	\item \underline{$\theta_i(t) - \theta_j (t) \in [-\pi,\pi]$:} According to lines 1-2 of Algorithm \ref{alg:calc-diff}: $$calcDiff(\theta_i(t),\theta_j(t)) = \theta_i(t) - \theta_j (t)$$
	Therefore, the proof for this case is derived instantly.
	\item \underline{$\theta_i(t) - \theta_j (t) < -\pi$:} According to lines 4-5 of Algorithm \ref{alg:calc-diff}: $$calcDiff(\theta_i(t),\theta_j(t)) = 2\pi + (\theta_i(t) - \theta_j (t)) < 2\pi - \pi = \pi$$
	where the inequality follows from our assumption in this case. We assume by contradiction that $calcDiff(\theta_i(t),\theta_j(t)) < -\pi$. Therefore:
	$$2\pi + (\theta_i(t) - \theta_j (t)) < -\pi \Rightarrow \theta_i(t) - \theta_j (t) < -3\pi$$
	Since $\theta_i(t),\theta_j(t) \in [0,2\pi)$, then: $\theta_i(t) - \theta_j (t) > 0 - 2\pi = -2\pi$.
	From the previous two inequalities we infer that $-2\pi < -3\pi$ - \underline{a contradiction}.
	\item \underline{$\theta_i(t) - \theta_j (t) > \pi$:} According to line 7 of Algorithm \ref{alg:calc-diff}: $$calcDiff(\theta_i(t),\theta_j(t)) = (\theta_i(t) - \theta_j (t)) - 2\pi > \pi - 2\pi = -\pi$$
	where the inequality follows from our assumption in this case. We assume by contradiction that $calcDiff(\theta_i(t),\theta_j(t)) > \pi$. Therefore:
	$$(\theta_i(t) - \theta_j (t)) - 2\pi > \pi \Rightarrow \theta_i(t) - \theta_j (t) > 3\pi$$
	Since $\theta_i(t),\theta_j(t) \in [0,2\pi)$, then: $\theta_i(t) - \theta_j (t) < 2\pi - 0 = 2\pi$.
	From the previous two inequalities we infer that $2\pi > 3\pi$ - \underline{a contradiction}.
	\end{enumerate}
\end{proof}

Moreover, since $\theta_i(t)$ is assumed to be within $[0,2\pi)$ for each agent $a_i$, the mathematical average orientation of all agents is restricted to be within $[0,2\pi)$ as well. Thus, we infer that for any agent $a_i$ and for any time step $t \geq 0$: $$\frac{1}{n_i(t)} \sum_{a_j \in N_i (t)}\theta_j(t) \in [0,2\pi)$$

In light of the above, we will make the following assumption for the simplification of the update rule defined by Equation \ref{eq:orientation-calcdiff}:
\begin{assumption}
\label{assump:update-rule-simplification}
	For each pair of agents $a_i,a_j$, the difference between their orientations is always within $[-\pi,\pi]$, i.e. $\theta_j(t) - \theta_i(t) \in [-\pi,\pi]$ for any time step $t \geq 0$.
\end{assumption}
For justifying this assumption, note that Algorithm \ref{alg:calc-diff} handles the special cases for the differences between the orientation of some agent $a_i$ and each of its neighbors, while restricting these differences such that they are always within $[-\pi,\pi]$ (According to Lemma \ref{lemma:difference-restricted}). Therefore, the assumption in question can be made for simplifying the update rule of each agent, while still restricting the mathematical average orientation of all agents to be within $[0, 2\pi)$ as previously mentioned.

Given this assumption, the resulting simplified update rule for each flocking agent $a_i$ is as follows:
\begin{equation}
\label{eq:simplified-calcdiff}
    \theta_i(t+1) = \frac{1}{n_i(t)} \sum_{a_j \in N_i (t)} \theta_j(t)
\end{equation}
The following lemma gives us this simplification in matrix terms:
\begin{lemma}
    \label{lemma:simple-calcdiff-matrix}
    Let the \textit{adjacency matrix} and the \textit{degree matrix} of the flocking neighbors graph at time step $t$ be $A(t),D(t)$ (respectively). Denoting $\theta(t) = [\theta_i(t)]$ ($0 \leq i \leq k-1$), the following heading system is equivalent to Equation \ref{eq:simplified-calcdiff}: $$\theta(t+1) = (I+D(t))^{-1}(I+A(t)) \theta(t)$$
\end{lemma}
\begin{proof}
    For each flocking agent $a_i$, its neighborhood is the set of $n_i(t) \leq k$ agents including agent $a_i$ at time step $t$. In contrast, the degree of the vertex $i$ in the flocking neighbors graph in time step $t$ equals to the number agents in agent $a_i$'s neighborhood \textbf{except for itself}, i.e., the degree of the vertex $i$ equals to $D(t)[i,i]=n_i(t) - 1$. Thus, it also holds that $A(t)[i,i]=0$ for each flocking agent $a_i$. 

    Therefore, for a pair of flocking agents $a_i,a_j$, the $(i,j)$-th cell of the matrix $(I+D(t))^{-1}(I+A(t))$ equals to $\frac{1}{n_i(t)}A(t)[i,j]$ and its $(i,i)$-th cell equals to $\frac{1}{n_i(t)}$. Therefore, for each flocking agent $a_i$ the following holds: $$((I+D(t))^{-1}(I+A(t)) \theta(t))[i]=$$
    $$= \sum_{i \neq j=0} ^{k-1}\frac{1}{n_i(t)}A(t)[i,j]\theta_j(t) + \frac{1}{n_i(t)}\theta_i(t) =$$ $$= \frac{1}{n_i(t)} \sum_{a_j \in N_i (t)} \theta_j(t) = \theta_i(t+1)$$
    where the second equality stems from the definition of the adjacency matrix $A(t)$ and the last one stems from Equation \ref{eq:simplified-calcdiff}.
\end{proof}

As a matter of fact, $P(t) := (I+D(t))^{-1}(I+A(t))$ is the \textbf{normalized} Perron matrix obtained from the following \textbf{normalized} Laplacian matrix with $\varepsilon=1$: 

\begin{equation}
    \label{eq:normalized-laplacian}
    \bar{L}(t) := I - (I+D(t))^{-1}(I+A(t))
\end{equation}

Lemma \ref{lemma:simple-calcdiff-matrix} that follows presents some structural properties regarding the normalized Perron matrix. To prove them, we will be using Gershgorin disk theorem \cite{gershgorin1931uber}, which is stated as follows:
\begin{theorem}
    \label{theorem:Gershgorin}
    (\textbf{Gershgorin Disk Theorem, 1931})
    Let $M$ be a complex $n \times n$ matrix. For $1 \leq i \leq n$, let $R_{i}=\sum _{{j\neq {i}}}\left|M[i,j]\right|$ be the sum of the absolute values of the non-diagonal entries in the $i$-th row. Let  $B(M[i,i],R_{i})\subseteq \mathbb {C}$ be a closed disc centered at $M[i,i]$ with radius $R_{i}$. Such a disc is called a \textbf{Gershgorin disc}. Then, Every eigenvalue of $M$ lies within at least one of the Gershgorin discs $B(M[i,i],R_{i})$.
\end{theorem}

Before proving these properties, we prove that both the following propety of the flocking neighbors graph and the influencing neighbors graph:
\begin{lemma}
\label{lemma:neighbors-graphs-balanced}
At each time step $t$, the flocking neighbors graph and the influencing neighbors graph are balanced digraphs (according to Definition \ref{def:balanced-digraphs}).
\end{lemma}
\begin{proof}
Since the proof for both graphs is identical, we will be proving the lemma only for the flocking neighbors. Let the \textit{adjacency matrix} of the flocking neighbors graph at time step $t$ be $A(t)$. Let $i \in \mathcal{V}$. For each $i \neq j \in \mathcal{V}$, we divide the proof into two cases:
\begin{enumerate}
    \item \underline{$a_j \in N_i(t)$ -} This also means that $a_i \in N_j(t)$. According to the definition of the flocking neighbors graph, we have that $A(t)[i,j] = A(t)[j,i] = 1$.
    \item \underline{$a_j \notin N_i(t)$ -} This also means that $a_i \notin N_j(t)$. According to the definition of the flocking neighbors graph, we have that $A(t)[i,j] = A(t)[j,i] = 0$.
\end{enumerate}
Thus: $\sum_{j \neq i} A(t)[i,j] = \sum_{j \neq i} A(t)[j,i]$, which means that the flocking neighbors graph is a balanced digraph.
\end{proof}
\underline{\textbf{Note:}} Let the \textit{adjacency matrix} of the flocking neighbors graph at time step $t$ be $A(t)$. If the flocking neighbors graph was assumed to be \textbf{undirected}, it could be directly inferred that $A(t)[i,j] = A(t)[j,i]$. In other words, each undirected graph is necessarily a balanced graph.

We shall also consider the following definitions:
\begin{definition}
\label{def:stochastic}
(\textit{Row Stochastic Matrix} \cite{horn1985johnson}) 
A matrix $M$ is called \textbf{row stochastic} if it is a \textbf{nonnegative} matrix (a matrix whose entries are all nonnegative) and its row sums all equal 1 (i.e., $M1 = 1$).
\end{definition}
\begin{definition}
\label{def:column-stochastic}
(\textit{Column Stochastic Matrix} \cite{horn1985johnson}) 
A matrix $M$ is called \textbf{column stochastic} if it is a \textbf{nonnegative} matrix (a matrix whose entries are all nonnegative) and its column sums all equal 1 (i.e., $1M = 1$).
\end{definition}
\begin{definition}
\label{def:doubly-stochastic}
(\textit{Doubly Stochastic Matrix}) 
A matrix $M$ is called \textbf{doubly stochastic} if it is a \textbf{nonnegative} matrix (a matrix whose entries are all nonnegative), which is both row stochastic and column stochastic.
\end{definition}
\begin{definition}
\label{def:irreducible}
(\textit{Irreducible Matrix}) 
An adjacency matrix $A(t)$ is \textbf{irreducible} if and only if $\mathcal{G}(t)$ is strongly connected.
\end{definition}
\begin{definition}
\label{def:period}
(\textit{Period}) 
Let $M$ be a nonnegative matrix. Fix an index $i$ and define the \textbf{period of index $i$} to be the greatest common divisor of all natural numbers $j$ such that $(M ^ j)[i,i] > 0$. When $M$ is irreducible, the period of every index is the same and is called the \textbf{period of $M$}. If the period is 1, A is \textbf{aperiodic}.
\end{definition}
\begin{definition}
\label{def:aperiodic}
(\textit{Aperiodic Matrix})
Let $M$ be an irreducible nonnegative matrix. If its period is 1, $M$ is \textbf{aperiodic}.
\end{definition}
\begin{definition}
\label{def:primitive}
(\textit{Primitive Matrix}) 
A matrix $M$ is \textbf{primitive} if it is an irreducible aperiodic nonnegative matrix.
\end{definition}
\begin{lemma} 
    \label{lemma:normalized-perron-properties}
    The normalized Perron matrix $P(t) = (I+D(t))^{-1}(I+A(t))$ has the following structural properties:
    \begin{enumerate}
        \item $P(t)$ is a \textbf{row stochastic nonnegative} matrix with a trivial eigenvalue of $1$.
        \item $P(t)$ is a \textbf{column stochastic nonnegative} matrix. Combined with the first property, this makes $P(t)$ a \textbf{doubly stochastic} matrix.
        \item If $\mathcal{G}(t)$ is a strongly connected graph, then $P(t)$ is a \textbf{primitive} matrix.
        \item All eigenvalues of $P(t)$ are in a unit disk.
    \end{enumerate}
\end{lemma}
\begin{proof}
The proof for each property is as follows:
    \begin{enumerate}
        \item Clearly, $P(t)$ is a nonnegative matrix. According to the definitions of $D(t)$ and $A(t)$:
        $$D(t)[i,j] = \begin{cases} n_i(t)-1, \quad i=j\\ 0, \quad i \neq j \end{cases} $$
        $$ A(t)[i,j] = \begin{cases} 1, \quad \{i,j\} \in \mathcal{E}(t)\\ 0, \quad otherwise \end{cases}$$
        Therefore, it holds that: 
        $$((I + A(t))1)[i,j] = (I + D(t))[i,j] = \begin{cases} n_i(t), \quad i=j\\ 0, \quad i \neq j \end{cases}$$
        which leads to $P(t)1 = 1$, i.e., $P(t)$ is \textbf{row stochastic}. 
        \item Since $G(t)$ is a balanced digraph (\ref{lemma:neighbors-graphs-balanced}), the proof is similar to the proof of the previous property.
        \item Given the stated above, we have that: 
        $$P(t)[i,j] = \frac{1}{n_i(t)} (I+A(t))[i,j] = \begin{cases} \frac{1}{n_i(t)}, \quad i=j\\ \frac{1}{n_i(t)} A(t)[i,j], \quad i \neq j \end{cases}$$
        That is, $P(t)$ is primitive if and only if $I + A(t)$ is primitive. Since $G(t)$ is a strongly connected digraph, according to Definition \ref{def:irreducible}, the matrix $I + A(t)$ is \textbf{irreducible}. Since $(I + A(t))[i,i] = 1 > 0$, then the period of $I + A(t)$ is 1, which means it is \textbf{aperiodic}. According to the first property, this matrix is also \textbf{nonnegative}, which means that $I + A(t)$ is a \textbf{primitive} matrix, as desired.
        \item Based on the Gershgorin disk theorem \cite{gershgorin1931uber}, all the eigenvalues of $P(t)$ are located in the union of the following disks: 
        $$\mathcal{B}_i(t) = \{z \in \mathbb{C}:\left| z - P(t)[i,i] \right| \leq \sum _{{j\neq {i}}}\left|P(t)[i,j]\right|\}$$
        According to the proof of the previous property: 
        $$\sum _{{j\neq {i}}}\left|P(t)[i,j]\right| = \sum _{{j\neq {i}}}\left|\frac{1}{n_i(t)} A(t)[i,j]\right| = \frac{n_i(t) - 1}{n_i(t)} = 1 - \frac{1}{n_i(t)}$$
        Therefore: 
        $$\mathcal{B}_i(t) = \Bigg\{z \in \mathbb{C}:\left| z - \frac{1}{n_i(t)} \right| \leq 1 - \frac{1}{n_i(t)}\Bigg\}$$
        Setting $j := \argmax_i n_i(t)$, clearly $\mathcal{B}_i(t) \subseteq \mathcal{B}_j(t)$. Since $1 - \frac{1}{n_i(t)} < 1$ for all $i$, then all the disks $\mathcal{B}_i(t)$ are proper subsets of the following unit disk:
        $$\Bigg\{z \in \mathbb{C}:\left| z - \frac{1}{n_j(t)} \right| \leq 1\Bigg\}$$
    \end{enumerate}
\end{proof}
Note that Lemma \ref{lemma:simple-calcdiff-matrix} also holds for the \textbf{influencing} neighbors graph (Definition \ref{def:influencing-neighbors-graph}). Denoting $\tilde{\theta}(t) = [\theta_i(t)]$ ($0 \leq i \leq n-1$), we have that:
\begin{equation}
\label{eq:influencing-heading-system}
    \tilde{\theta}(t+1) = (I+\tilde{D}(t))^{-1}(I+\tilde{A}(t)) \tilde{\theta}(t) =: \tilde{P}(t) \tilde{\theta}(t)
\end{equation}

Notice that this heading system induces a heading system for each connected component of the \textbf{influencing} neighbors graph at time step $0$. This heading system can be viewed while only considering the orientations of the agents which correspond to this specific connected component. We shall phrase this formally in the following lemma:
\begin{lemma}
    \label{lemma:connected-component-normalized-laplacian}
    Let $C$ be some connected component of $\tilde{\mathcal{G}}(0)$. We denote by $\tilde{\theta} ^ C (t)$ the orientations vector which corresponds to the agents in $C$. Then, $C$ induces a topology $\tilde{\mathcal{G}}(0)(C)$ (according to Definition \ref{def:induced-subgraph}), with a normalized Perron matrix $\tilde{P} ^ C(t)$ for which $\tilde{\theta} ^ C (t+1) = \tilde{P} ^ C(t) \tilde{\theta} ^ C (t)$ according to Equation \ref{eq:influencing-heading-system}.
\end{lemma}
\begin{proof}
    We denote by $n(C)$ the number of agents which correspond to $C$. Assuming that $a_{i_1}, \dots ,a_{i_{n(C)}}$ are those agents, we denote by $\tilde{\theta} ^ C (t) := [\theta_{i_1} (t+1), \dots, \theta_{i_{n(C)}} (t+1)]$ the orientations vector associated with the connected component $C$ at time step $t$. We denote by $\tilde{P} ^ C(t) \in \mathbb{R} ^ {n(C) \times n(C)}$ the Perron matrix which results after omitting all the rows and columns in $\tilde{P}(t)$ except for those that correspond to the indexes $i_1, \dots ,i_{n(C)}$. Clearly, we have that $\tilde{\theta} ^ C (t+1) = \tilde{P} ^ C(t) \tilde{\theta} ^ C (t)$ according to Equation \ref{eq:influencing-heading-system}.
\end{proof}

The convergence analysis of this discrete-time algorithm in Subsection \ref{sec:fixed-topology} relies on Perron-Frobenius Theorem \cite{horn1985johnson}, which is stated as follows:
\begin{theorem}
    \label{theorem:perron-frobenius}
    (\textbf{Perron-Frobenius Theorem, 1907/1912}) Let $M$ be a primitive nonnegative matrix. Then, $M$ has left and right eigenvectors $w$ and $v$ (respectively), satisfying $Mv = v, w^T M = w^T$, and $v^T w = 1$. In particular: $\lim_{t \rightarrow \infty} M^t = v w^T$.
\end{theorem}

\subsection{Fixed Topology}
\label{sec:fixed-topology}
In this section, we deal with the behaviour an influencing agent as well as its influence has upon a connected component in a \textbf{fixed topology}, after it is inserted into the neighborhood of a flocking agent which corresponds to some vertex in this connected component. According to Definition \ref{def:fixed-topology}, it holds that there exists some balanced digraph $\mathcal{G}=(\mathcal{V},\mathcal{E})$ such that $\mathcal{G}(t) = \mathcal{G}$ for any time step $t$. We denote this graph's adjacency matrix, degree matrix, Perron matrix and Laplacian matrix by $A,D,P,L$ (respectively). 

Similarly, considering the influencing neighbors graph, it holds that there exists some balanced digraph $\mathcal{\tilde{G}}=(\mathcal{\tilde{V}},\mathcal{\tilde{E}})$ such that $\mathcal{\tilde{G}}(t) = \mathcal{\tilde{G}}$ for any time step $t$. We denote this graph's adjacency matrix, degree matrix, Perron matrix and Laplacian matrix by $\tilde{A},\tilde{D},\tilde{P},\tilde{L}$ (respectively).

First, we shall describe the influence an influencing agent has upon some connected component, after it is inserted into the neighborhood of a flocking agent which corresponds to some vertex in this connected component. We will be considering each definition for the global orientation of both the flocking agents and the influencing agents.

\subsubsection{Discrete-Time Case}
\label{sec:dt-fixed}
\textbf{We will first consider Equation \ref{eq:orientation-perron}}. The following theorem gives us the orientation the influencing agents should adopt in order to make the flocking agents in some connected component converge to their desired orientation.
\begin{theorem}
\label{theorem:orientation-perron-influence-same-orient}
    Consider a \textbf{fixed topology} network $\mathcal{\tilde{G}}$ of agents, where the flocking agents update their global orientation according to Equation \ref{eq:orientation-perron}, with a Perron matrix $\tilde{P}= I - \varepsilon \tilde{P}$ and maximum degree $\Delta = \max_i(\sum_{j \neq i} \tilde{A}[i,j])$. Let $C$ be some connected component of $\mathcal{\tilde{G}}$. Let $\alpha_i$ be the orientation the flocking agents associated with this connected component should converge to. Let $\tilde{\theta}^C(0)$ be fixed.

    Assuming that $0 < \varepsilon < \frac{1}{\Delta}$ and $m(C)$ influencing agents are inserted into this connected component, \textbf{all} the influencing agents should \textbf{constantly} adopt the orientation $\alpha_i$ in order for the flocking agents to converge to $\alpha_i$.
\end{theorem}
\begin{proof}
    Considering Lemma \ref{lemma:connected-component-normalized-laplacian}, note that:
    $$\tilde{\theta}^C(t+1) = \tilde{P}^C \tilde{\theta}^C(t) = (\tilde{P}^C)^t \tilde{\theta}^C(0)$$
    Thus, a consensus is reached in discrete-time, if the limit $\lim_{t \rightarrow \infty} (\tilde{P}^C)^t$ exists. According to Perron-Frobenius Theorem, this limit exists for primitive matrices. 
    
    Therefore, we denote $\alpha := \lim_{t \rightarrow \infty} \tilde{\theta}^C_j(t)$, for a flocking agent $a_j \in A^F$, which corresponds to the vertex $j$ in the connected component $C$. 

    Considering Subsection \ref{sec:flocking-model}-model, we note that $k_j(0) \geq 1$ since $C$ is a connected component. Moreover, since the topology is fixed, $k_j(t)=k_j(0)$ and $m_j(t)=m_j(0)$ for any time step $t$. In particular, $N_j(0)=N_j(t)$ for any time step $t$. Necessarily, there exists a flocking agent for which some influencing agent lies within its neighborhood. Without loss of generality, we assume that $a_j$ is such a flocking agent, i.e., $m_j(t) \geq 1$. The following holds according to the definition of the Perron matrix $\tilde{P}^C$:
    $$\tilde{\theta}^C_j(t+1) = \tilde{\theta}^C_j(t) +$$ 
    $$+ \varepsilon \sum_{\ell_1 \in N_j(0) \cap A^F} (\tilde{\theta}^C_{\ell_1}(t) - \tilde{\theta}^C_j(t)) + \varepsilon \sum_{\ell_2 \in N_j(0) \cap A^D} (\tilde{\theta}^C_{\ell_2}(t) - \tilde{\theta}^C_j(t))$$
    Since all the influencing agents are \textbf{constantly} adopting the orientation $\alpha_i$, we have that:
    $$\tilde{\theta}^C_j(t+1) =$$
    $$= \tilde{\theta}^C_j(t) + \varepsilon \sum_{\ell_1 \in N_j(0) \cap A^F} (\tilde{\theta}^C_{\ell_1}(t) - \tilde{\theta}^C_j(t)) + \varepsilon \sum_{\ell_2 \in N_j(0) \cap A^D} (\alpha_i - \tilde{\theta}^C_j(t)) = $$ $$= \tilde{\theta}^C_j(t) + \varepsilon \sum_{\ell_1 \in N_j(0) \cap A^F} (\tilde{\theta}^C_{\ell_1}(t) - \tilde{\theta}^C_j(t)) + \varepsilon m_j(0) (\alpha_i - \tilde{\theta}^C_j(t))$$
    Taking the limit $t \rightarrow \infty$ in both sides results in the following:
    $$\alpha = \alpha + \varepsilon \sum_{\ell_1 \in N_j(0) \cap A^F} (\alpha - \alpha) + \varepsilon m_j(0) (\alpha_i - \alpha) \Rightarrow \boxed{\alpha = \alpha_i}$$

\end{proof}

The speed of reaching a consensus is the key in design of the network topology as well as analysis of the influence the influencing agents have upon the flocking agents in a given connected component. To demonstrate this in the discrete-time case, we note that the following quantity is invariant (while recalling that the Perron matrix $\tilde{P}$ is row stochastic - $1^T \tilde{P}=\tilde{P}$):
$$AVG_1(t+1) := \frac{1}{n}1^T \tilde{\theta}(t+1) = \frac{1}{n} (1^T \tilde{P}) \theta(t) = \frac{1}{n} 1^T \tilde{\theta}(t) = AVG_1(t)$$

This invariance allows to define the \textit{disagreement vector} \cite{olfati2004consensus}:

\begin{equation}
    \label{eq:dis-vec}
    \delta(t) = \tilde{\theta}(t) - AVG_1(t) 1
\end{equation}

Similarly to Lemma \ref{lemma:connected-component-normalized-laplacian}, Equation \ref{eq:dis-vec} induces a disagreement vector for each connected component of the flocking neighbors graph. We shall phrase this formally in the following lemma:
\begin{lemma}
    \label{lemma:induced-dis-vec-linear}
    Let $C$ be some connected component of $\mathcal{\tilde{G}}$. We denote by $\tilde{\theta} ^ C (t),\delta^C(t)$ the orientations vector and disagreement vector (respectively), which correspond to the flocking agents in $C$. We also denote $Avg^C(t) := \frac{1}{|C|}1^T \tilde{\theta}^C(t)$. Then, the connected component $C$ induces the disagreement vector $\delta^C(t) = \tilde{\theta}^C(t) - Avg^C(t) 1$.
\end{lemma}

Furthermore, the discrete-time dynamics of the disagreement vector is given as follows:
\begin{equation}
    \label{eq:dis-vec-dt}
    \delta(t+1) = \tilde{P} \delta(t)
\end{equation}

We make the following definitions:
\begin{definition}
\label{def:eigenvalues}
(\textit{Eigenvalues of a Matrix})
Let $\mathcal{M}$ be a matrix with real eigenvalues. We denote by $\lambda_1(\mathcal{M}) \leq \lambda_2(\mathcal{M}) \leq \ldots \leq \lambda_e(\mathcal{M})$ the eigenvalues of the matrix $\mathcal{M}$, ordered sequentially in an ascending order.
\end{definition}
\begin{definition}
\label{def:sym-part}
    (\textit{Symmetric Part of a Matrix})
    Let $\mathcal{M}$ be a matrix. We denote its \textit{symmetric part} by $\mathcal{M}^{Sym} = \frac{\mathcal{M}+\mathcal{M}^T}{2}$.
\end{definition}

Thus, based on \cite{godsil2013algebraic} and \textit{Theorem 3} from \cite{olfati2004consensus}, the following lemma gives us a property of second eigenvalue:

\begin{lemma}
\label{lemma:eigenvalues-sym}
Let $\mathcal{G}$ be an balanced digraph with Laplacian $L$ with a symmetric part $L^{Sym} = \frac{L+L^T}{2}$ and with Perron $P=I - \varepsilon L$ with a symmetric part $P^{Sym} = \frac{P+P^T}{2}$. Then, for all disagreement vectors $\delta$:
\begin{enumerate}
    \item $\lambda_2 = \min_{1^T\delta=0}\frac{\delta^T L \delta}{\delta^T \delta}$ with $\lambda_2 =\lambda_2(L^{Sym})$, i.e.: $$\delta^T L \delta \geq \lambda_2 ||\delta||^2$$
    \item $\mu_2 = \min_{1^T\delta=0}\frac{\delta^T P \delta}{\delta^T \delta}$ with $\mu_2 =1- \varepsilon \lambda_2$, i.e.: $$\delta^T P \delta \geq \mu_2 ||\delta||^2$$
\end{enumerate}
\end{lemma}
Now, we shall define the notion of exponential convergence:
\begin{definition}
\label{def:exp-conv}
(\textit{Exponential Convergence})
We say that $f(t)$ \textit{converges exponentially} toward $f^\infty$ with a rate $r < 0$ if $||f(t) - f^\infty|| \leq O( e^{rt})$.
\end{definition}

Using Lemma \ref{lemma:eigenvalues-sym}, \textit{Corollary 2} in \cite{olfati2004consensus} characterizes the performance of consensus. It is stated as follows:
\begin{lemma}
\label{lemma:dt-conv-rate}
Consider a \textbf{fixed topology} network $\mathcal{G}$ which is a strongly connected balanced digraph. Then, the following statements hold:
\begin{enumerate}
    \item The continuous-time disagreement vector $\delta$, as the solution for Equation \ref{eq:dis-vec-ct}, converges exponentially with a rate which equals to $\log(\mu_2)$, where $\mu_2 =\lambda_2(P^{Sym})$ with $P^{Sym} = \frac{P+P^T}{2}$, i.e.: 
    $$||\delta(t+1)|| \leq \mu_2||\delta(t)||$$
    \item $\Phi(\delta) = \delta(t)^T\delta(t)$ is a smooth, positive-definite and proper function, which acts as a valid Lyapunov function for the disagreement dynamics.
\end{enumerate}
\end{lemma}
Considering the first statement of this lemma, we infer that for any time step $t$ (by induction):
\begin{equation}
\label{eq:dis-vec-dt-vanish}
    ||\delta(t)|| \leq \mu_2 ^ t||\delta(0)||
\end{equation}

Following \textit{Lemma 3} from \cite{olfati2007consensus}, $P$ is a primitive matrix and $0 < \mu_2 < 1$. Thus, the vanishing of the disagreement vector is guaranteed. Moreover, considering the definition of exponential convergence in Definition \ref{def:exp-conv}, note that:
$$\mu_2 ^ t = e ^{t \log(\mu_2)}$$
Thus, the convergence rate equals to $\log(\mu_2)$, as stated in the first statement of the lemma. This result is consistent with the Definition \ref{def:exp-conv} since $\log(\mu_2) < 0$ due to $0 < \mu_2 < 1$.

Since $\theta(t) = \delta(t) + Avg(t) 1$, this also means that average-consensus is globally asymptotically reached with a speed faster or equal to  $\mu_2 =\lambda_2(P^{Sym})$ with $P^{Sym} = \frac{P+P^T}{2}$. Moreover, this lemma also holds for any connected component of the influencing neighbors graph $\mathcal{\tilde{G}}$ (while considering Lemma \ref{lemma:induced-dis-vec-linear}).

\textbf{We shall now consider Equation \ref{eq:orientation-calcdiff}}. The following lemma gives us the number of time steps it take for the flocking agents in some connected component to converge to their desired orientation.
\begin{lemma}
\label{lemma:orientation-calcdiff-influence}
    Consider a \textbf{fixed topology} network $\mathcal{G}$ of flocking agents given by Equation \ref{eq:orientation-calcdiff}. Let $C$ be some connected component of $\mathcal{G}$. Let $\alpha_i$ be the orientation the flocking agents associated with this connected component should converge to. Assuming $m(C)$ influencing agents are inserted into this connected component, they can influence the $|C|$ flocking agents to align to $\alpha_i$ within $Z = 1 + \Big\lceil \frac{\frac{\pi}{2}}{\frac{m(C)\pi}{|C|+m(C)}-\epsilon} \Big\rceil$ time steps.
\end{lemma}
\begin{proof}
    According to \textit{Lemma 4} in \cite{genter2013ad}, this can be done if $2\pi > (Z-2)\frac{m(C)\pi}{|C|+m(C)}-\epsilon$ and $\theta_i(t) - \alpha_i \leq (Z-1)\frac{m(C)\pi}{|C|+m(C)}-\epsilon + \beta_{ij}(t) + \theta_i(t) - \theta_i(t+1) + \pi$. Clearly, the first inequality always holds. For the second one, noting that $\beta_{ij}(t) \leq \pi$ by definition and that $\theta_i(t)$ is on both of its sides, it is equivalent to the inequality $\theta_i(t+1) - \alpha_i \leq (Z-1)\frac{m(C)\pi}{|C|+m(C)}-\epsilon + 2\pi$. Since $\theta_i(t+1) - \alpha_i \leq 2\pi$, this also holds.
\end{proof}

For characterizing the behaviour each influencing agent should adopt, we consider the simplification of the update rule defined by Equation \ref{eq:orientation-calcdiff}. Since we consider a fixed topology, according to Definition \ref{def:fixed-topology}, it holds that there exists some balanced digraph $\tilde{\mathcal{G}}=(\tilde{\mathcal{V}},\tilde{\mathcal{E}})$ such that $\tilde{\mathcal{G}}(t) = \tilde{\mathcal{G}}$ for any time step $t$. We denote this graph's adjacency matrix and degree matrix by $\tilde{A},\tilde{D}$ (respectively). Thus, by Equation \ref{eq:simplified-calcdiff}:
\begin{equation}
    \label{eq:fixed-simplified-calcdiff}
    \tilde{\theta}(t+1) = (I+\tilde{D})^{-1}(I+\tilde{A}) \tilde{\theta}(t) := \tilde{P} \tilde{\theta}(t)
\end{equation}

We would like to characterize the group decision value of a connected component in the influencing neighbors graph at time step $0$. The following lemma give us this desired characterization:

The following theorem gives us the orientation the influencing agents should adopt in order to make the flocking agents in some connected component converge to their desired orientation.
\begin{theorem}
\label{theorem:orientation-normalized-perron-influence}
    Consider a \textbf{fixed topology} network $\tilde{\mathcal{G}}$ of agents given by Equation \ref{eq:fixed-simplified-calcdiff}, with a Perron matrix $\tilde{P} = (I+\tilde{D})^{-1}(I+\tilde{A})$. Let $C$ be some connected component of $\tilde{\mathcal{G}}$. We denote by $|C|$ the number of agents which correspond to $C$ and assume that $a_{i_1}, \dots ,a_{i_{|C|}}$ are those agents. Let $\alpha_i$ be the orientation the agents associated with this connected component should converge to. Let $\theta(0)$ be fixed.

    Assuming that $m(C)$ influencing agents are inserted into this connected component, in order for the flocking agents to converge to $\alpha_i$, \textbf{all} the influencing agents should \textbf{constantly} adopt the orientation $\alpha_i$ in order for the flocking agents to converge to $\alpha_i$.
\end{theorem}
\begin{proof}
    Considering Lemma \ref{lemma:connected-component-normalized-laplacian}, note that:
    $$\tilde{\theta}^C(t+1) = \tilde{P}^C \tilde{\theta}^C(t) = (\tilde{P}^C)^t \tilde{\theta}^C(0)$$
    Thus, a consensus is reached in discrete-time, if the limit $\lim_{t \rightarrow \infty} (\tilde{P}^C)^t$ exists. By Lemma \ref{lemma:normalized-perron-properties}, $\tilde{P} ^ C$ is row stochastic, that is: $\tilde{P} ^ C 1 = 1$. By Perron-Frobenius Theorem, $\tilde{P} ^ C$ has a left eigenvector $w \in \mathbb{R} ^ {|C|}$ for which:
    \begin{enumerate}
        \item $w^T \tilde{P} ^ C = w^T$
        \item $1^T w = 1 \Rightarrow \sum_{j = 1} ^ {|C|} w_j = 1$
        \item $\lim_{t \rightarrow \infty} (\tilde{P} ^ C) ^t = 1 w^T$
    \end{enumerate} 
    
    Therefore, we denote $\alpha := \lim_{t \rightarrow \infty} \tilde{\theta}^C_j(t)$, for a flocking agent $a_j \in A^F$, which corresponds to the vertex $j$ in the connected component $C$.

    Considering Subsection \ref{sec:flocking-model}-model, we note that $k_j(0) \geq 1$ since $C$ is a connected component. Moreover, since the topology is fixed, $k_j(t)=k_j(0)$ and $m_j(t)=m_j(0)$ for any time step $t$. In particular, $N_j(0)=N_j(t)$ for any time step $t$. Necessarily, there exists a flocking agent for which some influencing agent lies within its neighborhood. Without loss of generality, we assume that $a_j$ is such a flocking agent, i.e., $m_j(t) \geq 1$. The following holds according to Equation \ref{eq:simplified-calcdiff} and Lemma \ref{lemma:simple-calcdiff-matrix}:
    $$\tilde{\theta}^C_j(t+1) = \frac{1}{n_j(0)} \sum_{\ell_1 \in N_j(0) \cap A^F} \tilde{\theta}^C_{\ell_1}(t) +  \frac{1}{n_j(0)} \sum_{\ell_2 \in N_j(0) \cap A^D} \tilde{\theta}^C_{\ell_2}(t)$$
    Since all the influencing agents are \textbf{constantly} adopting the orientation $\alpha_i$, we have that:
    $$\tilde{\theta}^C_j(t+1) = \frac{1}{n_j(0)} \sum_{\ell_1 \in N_j(0) \cap A^F} \tilde{\theta}^C_{\ell_1}(t) + \frac{1}{n_j(0)} \sum_{\ell_2 \in N_j(0) \cap A^D} \alpha_i = $$ $$= \frac{1}{n_j(0)} \sum_{\ell_1 \in N_j(0) \cap A^F} \tilde{\theta}^C_{\ell_1}(t) + \frac{m_j(0)}{n_j(0)} \alpha_i$$
    Taking the limit $t \rightarrow \infty$ in both sides results in the following:
    $$\alpha = \frac{1}{n_j(0)} \sum_{\ell_1 \in N_j(0) \cap A^F} \alpha + \frac{m_j(0)}{n_j(0)} \alpha_i \Rightarrow$$
    $$\Rightarrow n_j(0)\alpha = k_j(0)\alpha + m_j(0)\alpha_i \Rightarrow$$ $$\Rightarrow (n_j(0) - k_j(0))\alpha = m_j(0)\alpha_i \Rightarrow \boxed{\alpha = \alpha_i}$$
    where the last transition follows from the fact that $n_j(0) = k_j(0) + m_j(0)$.

\end{proof}

\subsubsection{Continuous-Time Case}
\label{sec:ct-fixed}
\textbf{We will first consider Equation \ref{eq:global-orient-linear}}. Similarly to Lemma \ref{lemma:connected-component-normalized-laplacian}, Equation \ref{eq:global-orient-linear} induces a heading system for each connected component of the influencing neighbors graph at time step $0$. We shall phrase this formally in the following lemma:
\begin{lemma}
    \label{lemma:connected-component-global-orient-linear}
    Let $C$ be some connected component of $\mathcal{\tilde{G}}(0)$. We denote by $\tilde{\theta} ^ C (t)$ the orientations vector which corresponds to the flocking agents in $C$. Then, $C$ induces a topology $\tilde{\mathcal{G}}(0)(C)$ (according to Definition \ref{def:induced-subgraph}), with a Laplacian matrix $\tilde{L} ^ C(t)$ for which $\dot{\tilde{\theta}} ^ C (t+1) = -\tilde{L} ^ C (t) \tilde{\theta} ^ C (t)$ according to Equation \ref{eq:global-orient-linear}.
\end{lemma}

The following theorem gives us the orientation the influencing agents should adopt in order to make the flocking agents in some connected component converge to their desired orientation. 

\begin{theorem}
\label{theorem:orientation-linear-influence-same-orient}
    Consider a \textbf{fixed topology} network $\mathcal{\tilde{G}}$ of agents (both flocking an influencing) given by Equation \ref{eq:global-orient-linear}, with a Laplacian matrix $\tilde{L}$. Let $C$ be some connected component of $\mathcal{\tilde{G}}$. Let $\alpha_i$ be the orientation the agents associated with this connected component should converge to. Let $\theta(0)$ be fixed.

    Assuming that $m(C)$ influencing agents are inserted into this connected component, \textbf{all} the influencing agents should \textbf{constantly} adopt the orientation $\alpha_i$ in order for the flocking agents to converge to $\alpha_i$.
\end{theorem}
\begin{proof}
    Considering Lemma \ref{lemma:connected-component-global-orient-linear}, note that the solution of the heading system induced by the connected component $C$ is as follows:
    $$\tilde{\theta}^C(t+1) = e^{- t \tilde{L}^C} \cdot \tilde{\theta}^C(0)$$
    Thus, a consensus is reached in discrete-time, if the limit $\lim_{t \rightarrow \infty} e^{- t \tilde{L}^C}$ exists. By \textit{Theorem 3} in \cite{olfati2004consensus}, this limit exists.
    
    Therefore, we denote $\alpha := \lim_{t \rightarrow \infty} \tilde{\theta}^C_j(t)$, for a flocking agent $a_j \in A^F$, which corresponds to the vertex $j$ in the connected component $C$. Clearly, this means that: $\lim_{t \rightarrow \infty} \dot{\tilde{\theta}}_j ^ C (t) = 0$.

    Considering Subsection \ref{sec:flocking-model}-model, we note that $k_j(0) \geq 1$ since $C$ is a connected component. Moreover, since the topology is fixed, $k_j(t)=k_j(0)$ and $m_j(t)=m_j(0)$ for any time step $t$. In particular, $N_j(0)=N_j(t)$ for any time step $t$. Necessarily, there exists a flocking agent for which some influencing agent lies within its neighborhood. Without loss of generality, we assume that $a_j$ is such a flocking agent, i.e., $m_j(t) \geq 1$. Let $\tilde{A} ^ C$ be the adjacency matrix which corresponds to the topology induced by the connected component $C$. The following holds according to Equation \ref{eq:global-orient-linear} and Lemma \ref{lemma:global-orient-linear-matrix}:
    $$\dot{\tilde{\theta}}_j ^ C (t+1) = \sum_{\ell_1 \in N_j(0) \cap A^F} \tilde{A} ^ C[j,\ell_1] (\tilde{\theta}^C_{\ell_1}(t) - \tilde{\theta}^C_j(t)) +$$
    $$ +\sum_{\ell_2 \in N_j(0) \cap A^D} \tilde{A} ^ C[j,\ell_2] (\tilde{\theta}^C_{\ell_2}(t) - \tilde{\theta}^C_j(t))$$
    Since all the influencing agents are \textbf{constantly} adopting the orientation $\alpha_i$, we have that:
    $$\dot{\tilde{\theta}}_j ^ C (t+1) = \sum_{\ell_1 \in N_j(0) \cap A^F} \tilde{A} ^ C[j,\ell_1] (\tilde{\theta}^C_{\ell_1}(t) - \tilde{\theta}^C_j(t)) +$$
    $$ + \sum_{\ell_2 \in N_j(0) \cap A^D} \tilde{A} ^ C[j,\ell_2] (\alpha_i - \tilde{\theta}^C_j(t)) =$$ $$= \sum_{\ell_1 \in N_j(0) \cap A^F} \tilde{A} ^ C[j,\ell_1] (\tilde{\theta}^C_{\ell_1}(t) - \tilde{\theta}^C_j(t)) + m_j(0)(\alpha_i - \tilde{\theta}^C_j(t))$$
    Taking the limit $t \rightarrow \infty$ in both sides results in the following:
    $$0 = \sum_{\ell_1 \in N_j(0) \cap A^F} \tilde{A} ^ C[j,\ell_1] (\alpha - \alpha) + m_j(0)(\alpha_i - \alpha) \Rightarrow \boxed{\alpha = \alpha_i}$$
    where the last transition follows from the fact that $m_j(0) \geq 1$.
\end{proof}

    

As discussed in Subsection \ref{sec:dt-fixed}, the speed of reaching a consensus is the key in design of the network topology as well as analysis of the influence the influencing agents have upon the flocking agents in a given connected component. For the sake of this analysis, we will be considering the disagreement vector presented in Equation \ref{eq:dis-vec}. The continuous-time dynamics of the disagreement vector is given as follows:
\begin{equation}
    \label{eq:dis-vec-ct}
    \dot{\delta}(t) = -\tilde{L} \delta(t)
\end{equation}

Using Lemma \ref{lemma:eigenvalues-sym}, \textit{Theorem 8} in \cite{olfati2004consensus} characterizes the performance of consensus. It is stated as follows:
\begin{lemma}
Consider a \textbf{fixed topology} network $\mathcal{G}$ which is a strongly connected digraph. Then, the following statements hold:
\begin{enumerate}
    \item The continuous-time disagreement vector $\delta$, as the solution for Equation \ref{eq:dis-vec-ct}, converges exponentially with a rate which equals to $-\lambda_2$, where $\lambda_2 =\lambda_2(L^{Sym})$ with $L^{Sym} = \frac{L+L^T}{2}$, i.e.: 
    $$||\delta(t)|| \leq ||\delta(0)||e^{-\lambda_2t}$$
    \item $V(\delta) = \frac{1}{2} ||\delta||^2$ is a smooth, positive-definite and proper function, which acts as a valid Lyapunov function for the disagreement dynamics.
\end{enumerate}
\end{lemma}
Since $\theta(t) = \delta(t) + Avg(t) 1$, the first statement of this lemma also means that (average-)consensus is globally asymptotically reached with a speed faster or equal to $\lambda_2 =\lambda_2(L^{Sym})$ with $L^{Sym} = \frac{L+L^T}{2}$. Moreover, this lemma also holds for any connected component of the influencing neighbors graph $\mathcal{\tilde{G}}$ (while considering Lemma \ref{lemma:induced-dis-vec-linear}).

The following lemma formally states the difference which arises between the discrete-time case and the continuous-time case:
\begin{lemma}
    \label{lemma:faster-conv}
    The continuous-time disagreement vector (given in Equation \ref{eq:dis-vec-ct}) converges exponentially \textbf{faster} than the continuous-time disagreement vector (given in Equation \ref{eq:dis-vec-dt-vanish}).
\end{lemma}
\begin{proof}
    According to Lemma \ref{lemma:dt-conv-rate}, the disagreement vector in the discrete-time case exponentially converges with a rate of $\log(\mu_2)$, whereas in the continuous-time case, it converges exponentially with a rate of $-\lambda_2$. Given that $\tilde{P} = I - \varepsilon \tilde{L}$, we have that $\mu_2 = 1- \varepsilon \lambda_2$, which means that  $\log(\mu_2) = \log(1- \varepsilon \lambda_2)$. Note that the following equivalences hold:
    $$-\lambda_2 < \log(1- \varepsilon \lambda_2) \iff e^{-\lambda_2} < 1- \varepsilon \lambda_2 \iff \varepsilon > \frac{e^{-\lambda_2} - 1}{\lambda_2}$$
    Since $e^{-\lambda_2} < 1$, then $\frac{e^{-\lambda_2} - 1}{\lambda_2} < 0$. Given that $\varepsilon > 0$, we have that the last inequality holds. This concludes the lemma.
\end{proof}

\textbf{We shall now consider Equation \ref{eq:global-orient-zero-one}}. Based on this equation, we note that the following quantity is invariant:
$$AVG_2(t) = \frac{\sum_i |N_i(t)|\theta_i(t)}{\sum_i |N_i(t)|}$$

Based on Lemma \ref{lemma:global-orient-zero-one-matrix}, we would like to characterize the group decision value of a connected component in the influencing neighbors graph. The following lemma give us this desired characterization:
\begin{lemma}
    \label{theorem:group-decision-influencing-zero-one}
    Let $C$ be some connected component of $\mathcal{\tilde{G}}$. We denote by $|C|$ the number of agents which correspond to $C$ and assume that $a_{i_1}, \dots ,a_{i_{|C|}}$ are those agents. Then, the group decision value is: $$\alpha = \sum_{j = 1} ^ {|C|} w_j \theta_{i_j}(0);\quad w_j = \frac{|N_{i_j}(t)|}{\sum_\ell |N_{i_\ell}(t)|}$$ 
\end{lemma}
\begin{proof}
     The connected component $C$ of $\mathcal{\tilde{G}}$ induces an strongly connected, balanced digraph. Therefore, according to \textit{Theorem 1} in \cite{olfati2007consensus}, the statement in the lemma follows.
\end{proof}
Given this lemma, the following theorem gives us the orientation the influencing agents should adopt in order to make the flocking agents in some connected component converge to their desired orientation.
\begin{lemma}
\label{lemma:orientation-zero-one-influence}
    Consider a \textbf{fixed topology} network $\mathcal{\tilde{G}}$ of agents given by Equation \ref{eq:global-orient-zero-one}. Let $C$ be some connected component of $\mathcal{\tilde{G}}$. We denote by $|C|$ the number of agents which correspond to $C$ and assume that $a_{i_1}, \dots ,a_{i_{|C|}}$ are those agents. Let $\alpha_i$ be the orientation the agents associated with this connected component should converge to. Let $\tilde{\theta}(0)$ be fixed.

    Assuming that $m(C)$ influencing agents are inserted into this connected component, \textbf{all} the influencing agents should \textbf{constantly} adopt the orientation $\alpha_i$ in order for the flocking agents to converge to $\alpha_i$.
\end{lemma}
\begin{proof}
    The proof is similar to the proof of Theorem \ref{theorem:orientation-normalized-perron-influence}.
\end{proof}

\subsection{Switching Topology}
\label{sec:switching-topology}
In this section, we deal with the influence an influencing agents has upon a connected component in a \textbf{switching topology}, after it is inserted into the neighborhood of a flocking agent which corresponds to some vertex in this connected component. The notion of a switching topology is formally defined in Definition \ref{def:switching-topolgy}. Clearly, for all $t \geq 0$, we denote the adjacency matrix, the degree matrix, the Laplacian matrix and the Perron matrix of the influencing neighbors graph $\mathcal{\tilde{G}}_{\sigma(t)}$ by $\tilde{A}_{\sigma(t)},\tilde{D}_{\sigma(t)},\tilde{L}_{\sigma(t)}, \tilde{P}_{\sigma(t)}$ (respectively).

As we noticed in Subsection \ref{sec:fixed-topology}, this assumption affects the influence an influencing agent has upon some connected component, after it is inserted into the neighborhood of a flocking agent which corresponds to some vertex in this connected component. We will be considering each definition for the global orientation of the flocking agents (Equations \ref{eq:orientation-perron} and \ref{eq:orientation-calcdiff}).

Our goal is showing that for a large class of switching signals and for any initial set of agent headings, the flocking agents converge to a disagreement. Let us consider two extreme cases:
\begin{enumerate}
	\item \underline{Convergence can never occur -} When the $i$th agent starts so far away from the rest, it never acquires any neighbors. Mathematically, this means not only that $\mathcal{\tilde{G}}_{\sigma(t)}$ is never connected for any time step $t \geq 0$, but also that vertex $i$ remains an isolated vertex for all $t$. We could encounter this situation when the visibility radius $R$ is very small.
	\item \ul{All the flocking agents remain neighbors of all the other flocking agents in their respective connected component for all time -} In this case, $\sigma$ would remain fixed along such a trajectory at that value in $q \in \mathcal{P}$ for which $\mathcal{\tilde{G}}_q$ has constant set of connected components. This is a similar situation to the fixed topology case.
\end{enumerate}

The situation of perhaps the greatest interest is between these two extremes, when $\mathcal{\tilde{G}}_{\sigma(t)}$ does not necessarily have a constant set of connected components for any $t \geq 0$, but when no strictly proper subset of the graph's vertices is isolated from the rest for all time. Establishing disagreement in this case is challenging because $\sigma$ changes with time and the heading systems according to Equations \ref{eq:orientation-perron} and \ref{eq:orientation-calcdiff} are not time-invariant. It is this case which we intend to study.

To do this, we shall introduce several concepts. This concepts deal with the connectivity between the agents across time and they are given in the following definitions.
\begin{definition}
\label{def:union-graphs}
(\textit{Union of a Collection of Simple Graphs}) The union of a collection of simple graphs $\{\mathcal{G}_{p_1},...,\mathcal{G}_{p_\gamma}\}$, each with vertex set $\mathcal{V}$, is the simple graph $\cup_{\lambda = 1}^{\gamma}\mathcal{G}_{p_\lambda}$ with vertex set $\mathcal{V}$ and edge set equaling $\cup_{\lambda = 1}^{\gamma}\mathcal{E}_{p_\lambda}$.
\end{definition}
\begin{definition}
\label{def:jointly-connected}
	(\textit{Jointly Connected}) We say that a collection $\{\mathcal{G}_{p_1},...,\mathcal{G}_{p_\gamma}\}$ is \textit{jointly connected} if the union of its members is a connected graph.
\end{definition}
\begin{definition}
\label{def:linked-together}
	(\textit{Linked Together Across a Time Interval}) We say that $s$ flocking agents are \textit{linked together across a time interval} $[t,\tau]$ if the collection $\{\mathcal{G}_{\sigma(t)},...,\mathcal{G}_{\sigma(\tau)}\}$ encountered along the interval, is jointly connected. 
\end{definition}

\subsubsection{Discrete-Time Case}
\label{sec:dt-switching}
\textbf{We will first consider Equation \ref{eq:orientation-perron}}. In this case, we consider the following discrete-time heading system for the \textbf{influencing} neighbors graph (which stems from Equation \ref{eq:orientation-perron}):
\begin{equation}
\label{eq:perron-switching}
    \tilde{\theta}(t+1) = \tilde{P}_{\sigma(t)} \tilde{\theta}(t) \quad ; \quad t \geq 0
\end{equation}

Similarly to Lemma \ref{lemma:connected-component-normalized-laplacian}, Equation \ref{eq:perron-switching} induces a heading system for each connected component of the flocking neighbors graph at time step $0$. We shall phrase this formally in the following lemma:
\begin{lemma}
    \label{lemma:connected-component-heading}
    Let $C$ be some connected component of $\mathcal{\tilde{G}}(0)$. We denote by $\tilde{\theta} ^ C (t)$ the orientations vector which corresponds to the flocking agents in $C$. Then, $C$ induces a topology $\tilde{\mathcal{G}}(0)(C)$ (according to Definition \ref{def:induced-subgraph}), with a Perron matrix $\tilde{P}_{\sigma(t)} ^ C$ for which $\tilde{\theta} ^ C (t+1) = \tilde{P}_{\sigma(t)} ^ C \tilde{\theta} ^ C (t)$ according to Equation \ref{eq:perron-switching}.
\end{lemma}

In essence, we would like to consider the conditions under which convergence is guaranteed for each connected component of the influencing neighbors graph at time step $0$. Based on \textit{Theorem 2} in \cite{jadbabaie2003coordination}, the following theorem gives us these conditions.
\begin{theorem}
    Let $C$ be some connected component of $\mathcal{\tilde{G}}(0)$. Let $\tilde{\theta}(0)$ be fixed and let $\sigma(t):\mathbb{N} \rightarrow \mathcal{Q}_n$ be a switching signal for which there exists an infinite sequence of contiguous, nonempty, bounded, time-intervals $[t_i,t_{i+1}]$, $i \geq 0$, starting at $t_0$, with the property that across each such interval, the flocking agents which correspond to $C$ are linked together. Then, alignment is asymptotically reached for the flocking agents which correspond to the connected component $C$. That is: $$\lim_{t \rightarrow \infty} \tilde{\theta} ^ C (t) = \alpha 1$$
\end{theorem}

Given some connected component $C$ in $\mathcal{G}(0)$, the solution can be explicitly expressed as follows:
\begin{equation}
    \label{eq:solution}
    \tilde{\theta} ^ C (t) = \Big(\prod_{\tau=0}^t \tilde{P}_{\sigma(\tau)} ^ C\Big) \tilde{\theta} ^ C (0)
\end{equation}

Therefore, the convergence of the heading system of $C$ as it is given in Lemma \ref{lemma:connected-component-heading} depends on whether the infinite product of non-negative stochastic primitive matrices $\tilde{P}_{\sigma(t)} ^ C, \dots, \tilde{P}_{\sigma(1)} ^ C, \tilde{P}_{\sigma(0)} ^ C$ has a limit. This problem has been studied by several mathematicians including Wolfowitz \cite{wolfowitz1963products}. The proof in \cite{jadbabaie2003coordination} relies on Wolfowitz's lemma:
\begin{lemma}
\label{lemma:wolfowitz}
    (\textbf{Wolfowitz, 1963}) Let $\mathcal{M} = \{M_1, \dots , M_g\}$ be a finite set of primitive stochastic matrices such that for any sequence of matrices $M_{q_1}, \dots , M_{q_h} \in \mathcal{M}$ with $h \geq 1$, the product $M_{q_1} \dots  M_{q_h}$ is a primitive matrix. Then, there exists a row vector $w$ such that: $$\lim_{h \rightarrow \infty}M_{q_1} \dots  M_{q_h} = 1w$$
\end{lemma}

The following lemma characterizes which orientation the influencing agents should constantly adopt for guaranteeing the convergence of the flocking agents to their desired orientation, while using Wolfowitz's lemma. 
\begin{theorem}
\label{theorem:orientation-perron-influence-switching-same-orient}
    Consider a \textbf{switching topology} network of agents given by Equation \ref{eq:perron-switching}, with a Perron matrix $\tilde{P}(t) = I - \tilde{\varepsilon} \tilde{L}(t)$ and maximum degree $\Delta(t) = \max_i(\sum_{j \neq i} \tilde{A}(t)[i,j])$. Let $C$ be some connected component of $\mathcal{\tilde{G}}(0)$. We denote by $|C|$ the number of agents which correspond to $C$ and assume that $a_{i_1}, \dots ,a_{i_{|C|}}$ are those agents. Let $\alpha_i$ be the orientation the flocking agents associated with this connected component should converge to. Let $\theta(0)$ be fixed and let $\sigma(t):\mathbb{N} \rightarrow \mathcal{Q}_n$ be a switching signal for which there exists an infinite sequence of contiguous, nonempty, bounded, time-intervals $[t_i,t_{i+1}]$, $i \geq 0$, starting at $t_0$, with the property that across each such interval, the agents which correspond to $C$ are linked together.

    Assuming that $m(C)$ influencing agents are inserted into this connected component, \textbf{all} the influencing agents should \textbf{constantly} adopt the orientation $\alpha_i$ in order for the flocking agents to converge to $\alpha_i$.
\end{theorem}
\begin{proof}
    According to Wolfowitz's lemma, we get that $\lim_{t \rightarrow \infty} \tilde{\theta} ^ C (t) = 1(w\tilde{\theta} ^ C(0)) = \alpha 1$ with $\alpha = w\tilde{\theta} ^ C(0)$. The vector $w$ depends on the switching sequence and cannot be determined a priori. In other words, the group decision is an undetermined quantity in the convex hull of all initial states. Thus, a consensus is reached.

    Since across each interval $[t_i,t_{i+1}]$, $i \geq 0$, the agents which correspond to $C$ are linked together, there exists a flocking agent $a_j \in A^F$, for which $\lim_{t \rightarrow \infty} m_j(t) > 0$. In other words, after consensus is reached, the flocking agent $a_j$ has at least one influencing agent in its neighborhood.

    Considering Subsection \ref{sec:flocking-model}-model, the following holds according to the definition of the Perron matrix $\tilde{P}^C_{\sigma(t)}$:
    $$\tilde{\theta}^C_j(t+1) = \tilde{\theta}^C_j(t) + \varepsilon \sum_{\ell_1 \in N_j(t) \cap A^F} (\tilde{\theta}^C_{\ell_1}(t) - \tilde{\theta}^C_j(t))+$$
    $$ + \varepsilon \sum_{\ell_2 \in N_j(t) \cap A^D} (\tilde{\theta}^C_{\ell_2}(t) - \tilde{\theta}^C_j(t))$$
    Since all the influencing agents are \textbf{constantly} adopting the orientation $\alpha_i$, we have that:
    $$\tilde{\theta}^C_j(t+1) = \tilde{\theta}^C_j(t) + \varepsilon \sum_{\ell_1 \in N_j(t) \cap A^F} (\tilde{\theta}^C_{\ell_1}(t) - \tilde{\theta}^C_j(t)) +$$
    $$+ \varepsilon \sum_{\ell_2 \in N_j(t) \cap A^D} (\alpha_i - \tilde{\theta}^C_j(t)) \Rightarrow$$
    $$\tilde{\theta}^C_j(t+1) = \tilde{\theta}^C_j(t) + \varepsilon \Bigg[\sum_{\ell_1 \in N_j(t) \cap A^F} \tilde{\theta}^C_{\ell_1}(t) - k_j(t)\tilde{\theta}^C_j(t) \Bigg] +$$
    \begin{equation}
    \label{eq:sum-perron}
    + \varepsilon m_j(t) (\alpha_i - \tilde{\theta}^C_j(t))
    \end{equation}
    
    Since $N_j(t) \cap A^F \subseteq A^F$ and $k_j(t) \leq k$, we have that:
    $$\sum_{\ell_1 \in N_j(t) \cap A^F} \tilde{\theta}^C_{\ell_1}(t) \leq \sum_{\ell_1 \in A^F} \tilde{\theta}^C_{\ell_1}(t) \quad,\quad k_j(t)\tilde{\theta}^C_j(t) \leq k\tilde{\theta}^C_j(t)$$
    Taking the limit $t \rightarrow \infty$ in both sides of both inequalities results in the following:
    $$\lim_{t \rightarrow \infty} \sum_{\ell_1 \in N_j(t) \cap A^F} \tilde{\theta}^C_{\ell_1}(t) \leq k\alpha \quad,\quad \lim_{t \rightarrow \infty} k_j(t)\tilde{\theta}^C_j(t) \leq k\alpha$$
    Subtracting between the second inequality from the first one, we have that:
    $$0 \leq \lim_{t \rightarrow \infty} \Bigg[\sum_{\ell_1 \in N_j(t) \cap A^F} \tilde{\theta}^C_{\ell_1}(t) - k_j(t)\tilde{\theta}^C_j(t)\Bigg] \leq 0$$
    According to the "Sandwich Theorem", we have that:
    $$\lim_{t \rightarrow \infty} \Bigg[\sum_{\ell_1 \in N_j(t) \cap A^F} \tilde{\theta}^C_{\ell_1}(t) - k_j(t)\tilde{\theta}^C_j(t)\Bigg] = 0$$
    Thus, taking the limit $t \rightarrow \infty$ in both sides of Equation \ref{eq:sum-perron} results in the following:
    $$\alpha = \alpha + 0 + \varepsilon \lim_{t \rightarrow \infty} m_j(t) (\alpha_i - \alpha) \Rightarrow \boxed{\alpha = \alpha_i}$$
    where the last transition follows from the fact that $\lim_{t \rightarrow \infty} m_j(t) > 0$.
\end{proof}

\textbf{We shall now consider Equation \ref{eq:orientation-calcdiff}}. In this case, we will be considering the simplification of the update rule for the influencing neighbors graph, defined by Equation \ref{eq:influencing-heading-system}. 

The following lemma characterizes which orientation the influencing agents should constantly adopt for guaranteeing the convergence of the flocking agents to their desired orientation, while using Wolfowitz's lemma. 
\begin{theorem}
\label{theorem:orientation-normalized-perron-influence-switching-same-orient}
    Consider a \textbf{switching topology} network of agents given by Equation \ref{eq:fixed-simplified-calcdiff}, with a Perron matrix $\tilde{P}(t) = (I + \tilde{D}(t)) ^ {-1} (I + \tilde{A}(t))$. Let $C$ be some connected component of $\mathcal{\tilde{G}}(0)$. We denote by $|C|$ the number of agents which correspond to $C$ and assume that $a_{i_1}, \dots ,a_{i_{|C|}}$ are those agents. Let $\alpha_i$ be the orientation the flocking agents associated with this connected component should converge to. Let $\theta(0)$ be fixed. 

    Let $\sigma(t):\mathbb{N} \rightarrow \mathcal{Q}_n$ be a switching signal for which there exists an infinite sequence of contiguous, nonempty, bounded, time-intervals $[t_i,t_{i+1}]$, $i \geq 0$, starting at $t_0$, with the property that across each such interval, the agents which correspond to $C$ are linked together.

    Assuming that $m(C)$ influencing agents are inserted into this connected component, \textbf{all} the influencing agents should \textbf{constantly} adopt the orientation $\alpha_i$ in order for the flocking agents to converge to $\alpha_i$.
\end{theorem}
\begin{proof}
    Given a connected component $C$ in $\tilde{\mathcal{G}}(0)$, according to Lemma \ref{lemma:normalized-perron-properties}, the normalized Perron matrix $\tilde{P}^C _{\sigma(t)}$ is a row stochastic primitive nonnegative matrix. Therefore, both Equation \ref{eq:solution} and Wolfowitz's lemma \cite{wolfowitz1963products} apply to this normalized Perron matrix. That is, we get that $\lim_{t \rightarrow \infty} \tilde{\theta} ^ C (t) = 1(w\tilde{\theta} ^ C(0)) = \alpha 1$ with $\alpha = w\theta ^ C(0)$ and $\lim_{t \rightarrow \infty}\prod_{\tau=0}^t \tilde{P}_{\sigma(\tau)} = 1w$. The vector $w$ depends on the switching sequence and cannot be determined a priori.

    Since across each interval $[t_i,t_{i+1}]$, $i \geq 0$, the agents which correspond to $C$ are linked together, there exists a flocking agent $a_j \in A^F$, for which $\lim_{t \rightarrow \infty} m_j(t) > 0$. In other words, after consensus is reached, the flocking agent $a_j$ has at least one influencing agent in its neighborhood.

    Considering Subsection \ref{sec:flocking-model}-model, the following holds according to the definition of the normalized Perron matrix: $$\tilde{\theta}^C_j(t+1) = \frac{1}{n_j(t)} \sum_{\ell_1 \in N_j(t) \cap A^F} \tilde{\theta}^C_{\ell_1}(t) +  \frac{1}{n_j(t)} \sum_{\ell_2 \in N_j(t) \cap A^D} \tilde{\theta}^C_{\ell_2}(t)$$
    Since all the influencing agents are \textbf{constantly} adopting the orientation $\alpha_i$, we have that:
    $$\tilde{\theta}^C_j(t+1) = \frac{1}{n_j(t)} \sum_{\ell_1 \in N_j(t) \cap A^F} \tilde{\theta}^C_{\ell_1}(t) + \frac{1}{n_j(t)} \sum_{\ell_2 \in N_j(t) \cap A^D} \alpha_i = $$ $$= \frac{1}{n_j(t)} \sum_{\ell_1 \in N_j(t) \cap A^F} \tilde{\theta}^C_{\ell_1}(t) + \frac{m_j(t)}{n_j(t)} \alpha_i \Rightarrow$$ $$n_j(t) \tilde{\theta}^C_j(t+1) = \sum_{\ell_1 \in N_j(t) \cap A^F} \tilde{\theta}^C_{\ell_1}(t) + m_j(t) \alpha_i \Rightarrow$$
    \begin{equation}
        \label{eq:sum-normalized-perron}
        m_j(t) \tilde{\theta}^C_j(t+1) + k_j(t) \tilde{\theta}^C_j(t+1) - \sum_{\ell_1 \in N_j(t) \cap A^F} \tilde{\theta}^C_{\ell_1}(t) = m_j(t) \alpha_i
    \end{equation}
    
    where the last transition follows from the fact that $n_j(t) = k_j(t)+m_j(t)$. Since $N_j(t) \cap A^F \subseteq A^F$ and $k_j(t) \leq k$, we have that:
    $$\sum_{\ell_1 \in N_j(t) \cap A^F} \tilde{\theta}^C_{\ell_1}(t) \leq \sum_{\ell_1 \in A^F} \tilde{\theta}^C_{\ell_1}(t) \quad,\quad k_j(t)\tilde{\theta}^C_j(t+1) \leq k\tilde{\theta}^C_j(t+1)$$
    Taking the limit $t \rightarrow \infty$ in both sides of both inequalities results in the following:
    $$\lim_{t \rightarrow \infty} \sum_{\ell_1 \in N_j(t) \cap A^F} \tilde{\theta}^C_{\ell_1}(t) \leq k\alpha \quad,\quad \lim_{t \rightarrow \infty} k_j(t)\tilde{\theta}^C_j(t+1) \leq k\alpha$$
    Subtracting between the first inequality from the second one, we have that:
    $$0 \leq \lim_{t \rightarrow \infty} \Bigg[k_j(t)\tilde{\theta}^C_j(t+1) - \sum_{\ell_1 \in N_j(t) \cap A^F} \tilde{\theta}^C_{\ell_1}(t)\Bigg] \leq 0$$
    According to the "Sandwich Theorem", we have that:
    $$\lim_{t \rightarrow \infty} \Bigg[k_j(t)\tilde{\theta}^C_j(t+1) - \sum_{\ell_1 \in N_j(t) \cap A^F} \tilde{\theta}^C_{\ell_1}(t)\Bigg] = 0$$
    Thus, taking the limit $t \rightarrow \infty$ in both sides of Equation \ref{eq:sum-normalized-perron} results in the following:
    $$\alpha \lim_{t \rightarrow \infty} m_j(t) + 0 = \alpha_i \lim_{t \rightarrow \infty} m_j(t) \Rightarrow \boxed{\alpha = \alpha_i}$$
    where the last transition follows from the fact that $\lim_{t \rightarrow \infty} m_j(t) > 0$.
\end{proof}

\subsubsection{Continuous-Time Case}
\label{sec:ct-switching}

\textbf{We will first consider Equation \ref{eq:global-orient-linear}}. In this case, we consider the following continuous-time heading system for the \textbf{influencing} neighbors graph (which stems from Lemma \ref{lemma:global-orient-linear-matrix}):
\begin{equation}
\label{eq:global-orient-linear-matrix-switching}
    \dot{\tilde{\theta}}(t) = -\tilde{L}_{\sigma(t)} \tilde{\theta}(t) \quad ; \quad t \geq 0
\end{equation}

Similarly to Equation \ref{eq:dis-vec-ct}, the continuous-time \textbf{switching} dynamics of the disagreement vector is given as follows:
\begin{equation}
    \label{eq:dis-vec-ct-switching}
    \dot{\delta}(t) = -\tilde{L}_{\sigma(t)} \tilde{\delta}(t)  \quad ; \quad t \geq 0
\end{equation}

When considering a connected component $C$ of $\mathcal{\tilde{G}}(0)$, we would like only to consider the entire set of strongly connected and balanced digraphs with $|C|$ vertices. Thus, we make the following definitions:
\begin{definition}
\label{def:L(G)}
    (\textit{$L(\mathcal{G})$})
    Given a graph $\mathcal{G}=(\mathcal{V},\mathcal{E})$, we denote its Laplacian matrix by $L(G)$.
\end{definition}
\begin{definition}
	(\textit{The Collection of All Strongly Connected And Balanced Digraphs With $s$ Vertices})
	$\Gamma_s = \{\mathcal{G}=(\mathcal{V},\mathcal{E})|rank(L(G))=s-1 \wedge 1^T L(G) = 0\}$.
\end{definition}

Note that $\Gamma_s$ is finite since at most a graph of order $s$ is complete and has $s(s-1)$ edges (there are $\binom{s}{2}=\frac{1}{2}s(s-1)$ distinct pairs of vertices). Therefore: $|\Gamma_s| \leq s(s-1)$. For a connected component $C$ of $\mathcal{\tilde{G}}(0)$, the subgraph induced by $C$ is clearly in $\Gamma_{|C|}$.

Similarly to Lemma \ref{lemma:connected-component-normalized-laplacian}, Equation \ref{eq:global-orient-linear-matrix-switching} induces a heading system for each connected component of the flocking neighbors graph at time step $0$. This case differs from the one introduced in Subsection \ref{sec:dt-switching} since we would like the switching signal for the heading system induced by a connected component $C$ of $\mathcal{\tilde{G}}(0)$ to be restricted as a $[0,\infty) \rightarrow \Gamma_{|C|}$ function. We shall phrase this formally in the following lemma:
\begin{lemma}
    \label{lemma:connected-component-heading-linear-switching}
    Let $C$ be some connected component of $\mathcal{\tilde{G}}(0)$. We denote by $\tilde{\theta} ^ C (t)$ and $\delta^C(t)$ the orientations vector and the disagreement vector (respectively), which corresponds to the agents in $C$. We assume that the induced subgraph $\mathcal{\tilde{G}}(t)(C)$ is strongly connected for any time step $t$. Then, the Laplacian matrix $\tilde{L}_{\sigma(t)}$ and the switching signal $\sigma(t)$ induce a Laplacian matrix $\tilde{L}_{\sigma(t)} ^ C$ and a switching signal $\sigma^C(t):[0,\infty) \rightarrow \Gamma_{|C|}$, for which: $$\dot{\tilde{\theta} }^ C (t+1) = -\tilde{L}_{\sigma^C(t)} ^ C \tilde{\theta} ^ C (t) \quad , \quad \dot{\delta} ^ C (t+1) = -\tilde{L}_{\sigma^C(t)} ^ C \delta ^ C (t)$$
\end{lemma}

Given this lemma and based on \textit{Theorem 9} in \cite{olfati2004consensus}, we have the following lemma:
\begin{lemma}
\label{lemma:linear-converge-switching-ct}
     Let $C$ be some connected component of $\mathcal{\tilde{G}}(0)$. We denote by $|C|$ the number of agents which correspond to $C$ and assume that $a_{i_1}, \dots ,a_{i_{|C|}}$ are those agents. We assume that the induced subgraph $\mathcal{\tilde{G}}(t)(C)$ is strongly connected for any time step $t$. Let $\sigma^C(t):[0,\infty) \rightarrow \Gamma_{|C|}$ be an arbitrary switching signal. Then, the following statements hold:
     \begin{enumerate}
         \item The switching heading system given in Lemma \ref{lemma:connected-component-heading-linear-switching} globally asymptotically converges to an average-consensus, i.e. the group decision value is: $$\alpha = \frac{1}{|C|}\sum_{j = 1} ^ {|C|} \theta_{i_j}(0)$$ 
         \item $V(\delta^C) = \frac{1}{2} ||\delta^C||^2$ is a smooth, positive-definite and proper function, which acts as a valid Lyapunov function for the disagreement dynamics of $\delta^C$. 
         \item While considering Definitions \ref{def:sym-part}, \ref{def:eigenvalues} and \ref{def:L(G)}, the disagreement vector vanishes exponentially with the least rate of:
         $$\lambda_2^* = \min_{\mathcal{G} \in \Gamma_{|C|}}\lambda_2(L(\mathcal{G})^{Sym})$$
         i.e., the inequality $||\delta^C(t)|| \leq ||\delta^C(0)||e^{-\lambda_2^* t}$ holds.
     \end{enumerate}
\end{lemma}

Considering the third statement in the previous lemma, the following lemma gives us a \textbf{faster} speed rate for the vanishing of the disagreement vector:
\begin{lemma}
    Let $C$ be some connected component of $\mathcal{\tilde{G}}(0)$. We assume that the induced subgraph $\mathcal{\tilde{G}}(t)(C)$ is strongly connected for any time step $t$. Recalling that $\mathcal{G}_{\sigma(t)}(C)$ is the subgraph of $\mathcal{G}_{\sigma(t)}$ induced by $C$, the disagreement vector vanishes exponentially with the least rate of:
    $$\lambda_2^{**} = \min_{t \in [0,\infty)}\lambda_2(L(\mathcal{G}_{\sigma(t)}(C))^{Sym})$$
    i.e., the inequality $||\delta^C(t)|| \leq ||\delta^C(0)||e^{-\lambda_2^{**} t}$ holds. Moreover, this speed rate is \textbf{faster} then the one presented in the previous lemma.
\end{lemma}
\begin{proof}
    For some connected component $C$ of $\mathcal{\tilde{G}}(0)$, calculating $\dot{V}(\delta^C)$ as it is given in the second statement of Lemma \ref{lemma:linear-converge-switching-ct}, for all $\delta^C \neq 0$, we have that:
    $$\dot{V}(\delta^C) = -(\delta^C)^T L(\mathcal{G}_{\sigma(t)}(C)) \delta^C = -(\delta^C)^T L(\mathcal{G}_{\sigma(t)}(C))^{Sym} \delta^C \leq $$ $$\leq -\lambda_2(L(\mathcal{G}_{\sigma(t)}(C))^{Sym}) ||\delta^C||^2 \leq -\lambda_2^{**}||\delta^C||^2 = -2\lambda_2^{**}V(\delta^C) < 0$$
    Therefore: $\dot{V}(\delta^C) \leq -2\lambda_2^{**}V(\delta^C)$. Integrating both sides of the inequality gives us:
    $$V(\delta^C(t)) \leq V(\delta^C(0)) e^{-2\lambda_2^{**}t} \Rightarrow ||\delta^C(t)|| \leq ||\delta^C(0)||e^{-\lambda_2^{**}t}$$
    Note that: 
    $$\{\mathcal{G}_{\sigma(t)}(C)|t \in [0,\infty)\} \subseteq \Gamma_{|C|}$$
    Then, $\lambda_2^{**}$ always exists since $\Gamma_{|C|}$ is finite and $\lambda_2^{**} \leq \lambda_2^*$, meaning that the value of $\lambda_2^{**}$ gives a more strict bound on the vanishing of the disagreement vector $\delta$. In other words, it is a \textbf{faster} speed rate.
\end{proof}

\textbf{\underline{Note:}} Throughout all the lemmas in this section until this point, we assumed that the subgraph induced by each connected component is strongly connected for any time step $t$. This assumption gave us the guarantee that an average-consensus is reached. Generally, this is quite a restrictive condition since it does not apply to most topologies. In essence, we would like to consider the conditions under which convergence is still guaranteed for each connected component of the influencing neighbors graph at time step $0$, even when this assumption does not hold.

As discussed in \cite{jadbabaie2003coordination}, \cite{ren2005consensus}, \cite{morse1996supervisory}, such continuous time control laws might lead to chattering since neighbor relations can suddenly change, while the agents' positions also change. A solution to this problem is similar to what is described in \cite{cao2010dwell}, which introduces \textbf{dwell time}, that is, the Laplacian matrix $\tilde{L}_{\sigma(t)}$ is piecewise constant. In our context, this means that each agent is constrained to changing its control law only at discrete time, that is, the Laplacian matrix $\tilde{L}_{\sigma(t)}$ is piecewise constant. This avoids the chattering problem. Thus, Equation \ref{eq:global-orient-linear-matrix-switching} can be rewritten as:
\begin{equation}
\label{eq:global-orient-linear-matrix-switching-dwell}
    \dot{\tilde{\theta}}(t) = -\tilde{L}_{\sigma(t_i)} \tilde{\theta}(t) \quad ; \quad t \in [t_i,t_i+\tau_i]
\end{equation}

where $t_0$ is the initial time and $t_1,t_2,\dots$ is an infinite time sequence at which the influencing neighbors graph change, resulting in a change in $\tilde{L}_{\sigma(t_i)}$. Let $\tau_i = t_{i+1} -t_i$ be the \textbf{dwell time}, $i \geq 0$. Note that the solution to Equation \ref{eq:global-orient-linear-matrix-switching-dwell} is given by (after integrating both sides of the Equation \ref{eq:global-orient-linear-matrix-switching-dwell}):
\begin{equation}
    \label{eq:product-exp}
    \tilde{\theta}(t) = e^{-\tilde{L}_{\sigma(t_\ell)}(t-t_\ell)}\Big(\prod_{i=0}^{\ell-1}e^{-\tilde{L}_{\sigma(t_i)}(\tau_i)}\Big)\tilde{\theta}(0)
\end{equation}

where $\ell$ is the largest nonnegative integer for which $t_\ell \leq t$. \textbf{We will assume that the dwell time is constant $\tau_i = \tau$ for all $i$}. Note that we are also assuming that the agents are all synchronized by specifying the same infinite time sequence for all the agents. This means that switching signal is piecewise constant with successive switching times separated by $\tau$ time units. Based on \textit{Theorem 5} in \cite{jadbabaie2003coordination}, we infer the following theorem:
\begin{theorem}
    Let $C$ be some connected component of $\mathcal{\tilde{G}}(0)$. Let $\tilde{\theta} (0)$, $t_0$ and $\tau > 0$ be fixed. Let $\sigma(t):\mathbb{N} \rightarrow \mathcal{Q}_n$ be a piecewise constant switching signal whose switching times $t_1,t_2,\dots$ satisfy $t_{i+1}-t_i \geq \tau$, $i \geq 1$. We assume that the infinite sequence of contiguous, nonempty, bounded, nonoverlapping time-intervals $[t_{i_j},t_{i_j+\ell_j})$, $j \geq 1$, has the property that across each such interval, the agents which correspond to $C$ are linked together. Then, alignment is asymptotically reached for the agents which correspond to the connected component $C$. That is: $$\lim_{t \rightarrow \infty} \tilde{\theta} ^ C (t) = \alpha 1$$
\end{theorem}
For analyzing the consensus value, note that by Equation \ref{eq:product-exp} the convergence depends on the product $e^{-\tilde{L}_{\sigma(t_\ell)}(t-t_\ell)}\Big(\prod_{i=0}^{\ell-1}e^{-\tilde{L}_{\sigma(t_i)}(\tau_i)}\Big)$. The following lemma, which is based on \textit{Lemma 3.11} in \cite{ren2005consensus}, show the case in which this product is a primitive matrix:
\begin{lemma}
\label{lemma:exp-product-primitive}
Let $t_0$ be the initial time and $t_1,t_2,\dots$ be an infinite time sequence. Let $\tau_i = t_{i+1} -t_i >  0$, $i \geq 0$, and assume they are bounded. If the union of the digraphs $\{\mathcal{G}_{i_1},\dots,\mathcal{G}_{i_j}\} \subseteq \Gamma_s$ is strongly connected, then the product $e^{-\tilde{L}_{\sigma(t_\ell)}(t-t_\ell)}\Big(\prod_{i=0}^{\ell-1}e^{-\tilde{L}_{\sigma(t_i)}(\tau_i)}\Big)$ is a primitive matrix.
\end{lemma}

The following lemma characterizes which orientation the influencing agents should constantly adopt for guaranteeing the convergence of the flocking agents to their desired orientation, while using this lemma and Wolfowitz's lemma.
\begin{theorem}
\label{theorem:orientation-linear-influence-switching}
    Consider a \textbf{switching topology} network of agents given by Equation \ref{eq:global-orient-linear}. Let $C$ be some connected component of $\mathcal{\tilde{G}}(0)$. Let $\tilde{\theta} (0)$, $t_0$ and $\tau > 0$ be fixed. Let $\alpha_i$ be the orientation the agents associated with this connected component should converge to. 

    Let $\sigma(t):\mathbb{N} \rightarrow \mathcal{Q}_n$ be a piecewise constant switching signal whose switching times $t_1,t_2,\dots$ satisfy $t_{i+1}-t_i \geq \tau$, $i \geq 1$. We assume that the infinite sequence of contiguous, nonempty, bounded, nonoverlapping time-intervals $[t_{i_j},t_{i_j+\ell_j})$, $j \geq 1$, has the property that across each such interval, the agents which correspond to $C$ are linked together.

    Assuming that $m(C)$ influencing agents are inserted into this connected component, \textbf{all} the influencing agents should \textbf{constantly} adopt the orientation $\alpha_i$ in order for the flocking agents to converge to $\alpha_i$.
\end{theorem}
\begin{proof}
   According to Wolfowitz's lemma and Lemma \ref{lemma:exp-product-primitive}, there exists a row vector $w$ such that:
   $$\lim_{t \rightarrow \infty}e^{-\tilde{L}_{\sigma(t_\ell)}(t-t_\ell)}\Big(\prod_{i=0}^{\ell-1}e^{-\tilde{L}_{\sigma(t_i)}(\tau_i)}\Big) = 1w$$
   We get that $\lim_{t \rightarrow \infty} \tilde{\theta} ^ C (t) = 1(w\tilde{\theta} ^ C(0)) = \alpha 1$ with $\alpha = w\tilde{\theta} ^ C(0)$. The vector $w$ depends on the switching sequence and cannot be determined a priori. In other words, the group decision is an undetermined quantity in the convex hull of all initial states. Clearly, this means that: $\lim_{t \rightarrow \infty} \dot{\tilde{\theta}} ^ C (t) = 0$.

   Since across each interval $[t_{i_j},t_{i_j+\ell_j})$, $j \geq 1$, the agents which correspond to $C$ are linked together, there exists a flocking agent $a_{j_1} \in A^F$, for which $\lim_{t \rightarrow \infty} m_j(t) > 0$. In other words, after consensus is reached, the flocking agent $a_j$ has at least one influencing agent in its neighborhood.

    Considering Subsection \ref{sec:flocking-model}-model-model, the following holds according to the definition of the Laplacian matrix: $$\dot{\tilde{\theta}}_j ^ C(t+1) = \sum_{\ell_1 \in N_j(t) \cap A^F} \tilde{A}_{\sigma(t)} ^ C[j,\ell_1] (\tilde{\theta}^C_{\ell_1}(t) - \tilde{\theta}^C_j(t)) +$$
    $$+ \sum_{\ell_2 \in N_j(t) \cap A^D} \tilde{A}_{\sigma(t)} ^ C[j,\ell_2] (\tilde{\theta}^C_{\ell_2}(t) - \tilde{\theta}^C_j(t))$$
    Since all the influencing agents are \textbf{constantly} adopting the orientation $\alpha_i$, we have that:
    $$\dot{\tilde{\theta}}_j ^ C (t+1) = \sum_{\ell_1 \in N_j(0) \cap A^F} \tilde{A}_{\sigma(t)} ^ C[j,\ell_1] (\tilde{\theta}^C_{\ell_1}(t) - \tilde{\theta}^C_j(t)) +$$
    $$+ \sum_{\ell_2 \in N_j(0) \cap A^D} \tilde{A}_{\sigma(t)} ^ C[j,\ell_2] (\alpha_i - \tilde{\theta}^C_j(t)) =$$ 
    \begin{equation}
    \label{eq:laplacian-switching-sum}
        = \sum_{\ell_1 \in N_j(t) \cap A^F} \tilde{A}_{\sigma(t)} ^ C[j,\ell_1] \tilde{\theta}^C_{\ell_1}(t) - k_j(t) \tilde{\theta}^C_j(t) + m_j(t)(\alpha_i - \tilde{\theta}^C_j(t))
    \end{equation}

    Since $N_j(t) \cap A^F \subseteq A^F$, $\tilde{A}_{\sigma(t)} ^ C[j,\ell_1] \leq 1$ and $k_j(t) \leq k$, we have that:
    $$\sum_{\ell_1 \in N_j(t) \cap A^F} \tilde{A}_{\sigma(t)} ^ C[j,\ell_1] \tilde{\theta}^C_{\ell_1}(t) \leq \sum_{\ell_1 \in A^F} \tilde{\theta}^C_{\ell_1}(t) \quad,\quad k_j(t)\tilde{\theta}^C_j(t) \leq k\tilde{\theta}^C_j(t)$$
    Taking the limit $t \rightarrow \infty$ in both sides of both inequalities results in the following:
    $$\lim_{t \rightarrow \infty} \sum_{\ell_1 \in N_j(t) \cap A^F} \tilde{A}_{\sigma(t)} ^ C[j,\ell_1] \tilde{\theta}^C_{\ell_1}(t) \leq k\alpha \quad,\quad \lim_{t \rightarrow \infty} k_j(t)\tilde{\theta}^C_j(t+1) \leq k\alpha$$
    Subtracting between the second inequality from the first one, we have that:
    $$0 \leq \lim_{t \rightarrow \infty} \Bigg[\sum_{\ell_1 \in N_j(t) \cap A^F} \tilde{A}_{\sigma(t)} ^ C[j,\ell_1] \tilde{\theta}^C_{\ell_1}(t) - k_j(t)\tilde{\theta}^C_j(t+1)\Bigg] \leq 0$$
    According to the "Sandwich Theorem", we have that:
    $$\lim_{t \rightarrow \infty} \Bigg[\sum_{\ell_1 \in N_j(t) \cap A^F} \tilde{A}_{\sigma(t)} ^ C[j,\ell_1] \tilde{\theta}^C_{\ell_1}(t) - k_j(t)\tilde{\theta}^C_j(t+1)\Bigg] = 0$$
    Thus, taking the limit $t \rightarrow \infty$ in both sides of Equation \ref{eq:laplacian-switching-sum} results in the following:
    $$0 = 0 + \lim_{t \rightarrow \infty} m_j(t) (\alpha_i - \alpha)  \Rightarrow \boxed{\alpha = \alpha_i}$$
    where the last transition follows from the fact that $\lim_{t \rightarrow \infty} m_j(t) > 0$.
\end{proof}

\textbf{We shall now consider Equation \ref{eq:global-orient-zero-one}}. As noted in Lemma \ref{lemma:global-orient-zero-one-matrix}, this is a special case of an influencing neighbors graph $\mathcal{G}^*(t)$, for which $I$ is the degree matrix and $\tilde{D}(t) ^ {-1} \tilde{A}(t)$ is the adjacency matrix, i.e., its Laplacian matrix equals to $I - \tilde{D}(t) ^ {-1} \tilde{A}(t)$. Thus, the convergence analysis is identical to the previous analysis in the switching topology case.

\section{The Intersection Points Placement Method}
\label{sec:intersection-method}
Thus far, we have discussed the influence that one or more influencing agents have upon a single component, once they are initially placed in a manner that they can act to do so. In this section, we propose the \textit{Intersection Points Placement Method}, which guarantees that a \textit{single} influencing agent's initial placement is such that it will necessarily influence the entire connected component.

Beforehand, in the following lemma we formally prove the mathematical expression for the locus of all the points in $\mathbb{R}^2$, which lie on this linear line. Throughout the entire proof, we consider Figure \ref{fig:circles-intersection}.
\begin{Figure}
    \captionof{figure}{An illustration for the intersection of two closed discs, utilized in Lemma
    \ref{lemma:discs-intersection}.}
    \label{fig:circles-intersection}
    \centering
    \includegraphics[width=150px,height=150px]{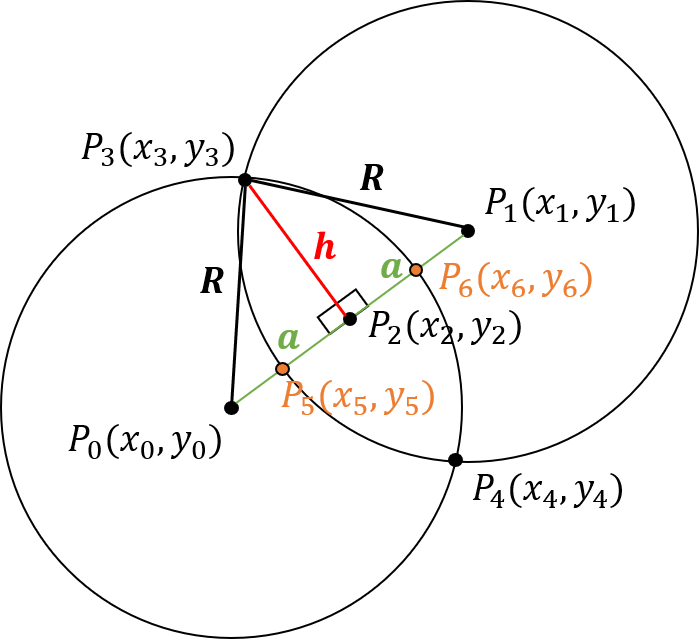}
\end{Figure}

\begin{lemma}
\label{lemma:discs-intersection}
Let us consider a pair of discs $\mathcal{D}_i = \{(x,y)\in \mathbb{R}^2|(x-x_i)^2 + (y-y_i)^2 \leq R^2\}$. We denote the distance between their centers by $d > 0$. Then, $\mathcal{D}_0 \cap \mathcal{D}_1 \neq \emptyset$ if and only if $0 < d \leq 2R$, and:
$$x_{3,4} = \frac{x_0+x_1}{2} \mp \frac{y_0-y_1}{2d}\sqrt{4R^2 - d^2}$$
$$y_{3,4} = \frac{y_0+y_1}{2} \pm \frac{x_0-x_1}{2d}\sqrt{4R^2 - d^2}$$
Consequently, the locus $\mathcal{L} \subseteq \mathcal{D}_0 \cap \mathcal{D}_1$ of all points $(x,y) \in \mathbb{R}^2$, which lie on the linear line connecting the intersection points $P_3,P_4$, is given by the following:
\begin{equation}
    \label{eq:locus-intersection}
    \mathcal{L} = \Big\{(x,y) \in \mathbb{R}^2 \Big| y = \frac{x_1-x_0}{y_0-y_1}x + \frac{x_1^2 - x_0^2 + y_0^2 - y_1^2}{2(x_0-x_1)} \wedge x \in [x_3,x_4]\Big\}
\end{equation}
\end{lemma}
\begin{proof}
   We denote by $P_i (x_i,y_i)$ the center of the disc $\mathcal{D}_i$.
   
   Considering the isosceles triangle $\bigtriangleup P_0 P_1 P_3$, the perpendicular bisector $P_2 P_3$ of the base is also its median, that is $a := |P_0 P_2| = |P_1 P_2|$ and $P_2(\frac{x_0+x_1}{2},\frac{y_0+y_1}{2})$. Clearly:
   $|P_0 P_1| = |P_0 P_2| + |P_1 P_2| \Rightarrow d = 2a \Rightarrow a = \frac{d}{2}$. According to the Pythagorean theorem in the right-angled triangle $\bigtriangleup P_0 P_2 P_3$, we have that: $a^2 + h^2 = R^2 \Rightarrow h^2 = R^2-a^2 \Rightarrow h^2 = R^2-\frac{d^2}{4}$.

   Considering the distance between the points $P_2$ and $P_3$, we have that:
   \begin{equation}
       \label{eq:dist-center-p3}
       \Big(x_3 - \frac{x_0+x_1}{2}\Big)^2 + \Big(y_3 - \frac{y_0+y_1}{2}\Big)^2 = h^2
   \end{equation}
   We infer from Equation \ref{eq:dist-center-p3} that:
   $$\Big(x_3 - x_0 + \frac{x_0-x_1}{2}\Big)^2 + \Big(y_3 - y_0 + \frac{y_0-y_1}{2}\Big)^2 = h^2 \Rightarrow$$
   $$\Rightarrow (x_3 - x_0)^2 + 2(x_3 - x_0)\frac{x_0-x_1}{2} + \frac{(x_0-x_1)^2}{4} +$$
   $$+ (y_3 - y_0)^2 + 2(y_3 - y_0)\frac{y_0-y_1}{2} + \frac{(y_0-y_1)^2}{4} = h^2 \Rightarrow$$
   $$\Rightarrow (x_3 - x_0)^2 + (y_3 - y_0)^2 + \frac{(x_0-x_1)^2 + (y_0-y_1)^2}{4} +$$
   $$+ (x_3 - x_0)(x_0-x_1) + (y_3 - y_0)(y_0-y_1) = h^2$$
   Noting that $(x_3 - x_0)^2 + (y_3 - y_0)^2 = R^2$, $(x_0-x_1)^2 + (y_0-y_1)^2 = d^2$ and substituting the expression for $h^2$, we have that:
   $$(x_3 - x_0)(x_0-x_1) + y_3(y_0-y_1) - y_0(y_0-y_1) + R^2 + \frac{d^2}{4}= R^2 - \frac{d^2}{4} \Rightarrow$$
   \begin{equation}
       \label{eq:y-3-4}
       \Rightarrow y_3 = \frac{x_0-x_1}{y_0-y_1}(x_0-x_3) + y_0 -\frac{d^2}{2(y_0-y_1)}
   \end{equation}
   Since the point $P_3$ lies on the exterior of the disc $\mathcal{C}_0$ (while denoting $z := x_0-x_3$):
   $$(x_3-x_0)^2 + (y_3-y_0)^2=R^2 \Rightarrow$$
   $$\Rightarrow z^2 + \Bigg(\frac{x_0-x_1}{y_0-y_1}z -\frac{d^2}{2(y_0-y_1)}\Bigg)^2 = R^2 \Rightarrow$$
   $$\Rightarrow z^2 + \Bigg(\frac{x_0-x_1}{y_0-y_1}\Bigg)^2 z^2 -2\frac{d^2(x_0-x_1)}{2(y_0-y_1)^2}z + \frac{d^4}{4(y_0-y_1)^2} = R^2 \Rightarrow$$
   $$\Rightarrow \frac{(x_0-x_1)^2 + (y_0-y_1)^2}{(y_0-y_1)^2}z^2 - \frac{d^2(x_0-x_1)}{(y_0-y_1)^2}z + \frac{d^4}{4(y_0-y_1)^2} -R^2 = 0$$
   Noting that $(x_0-x_1)^2 + (y_0-y_1)^2 = d^2$:
   \begin{equation}
       \label{eq:quadratic-intersect}
       \frac{d^2}{(y_0-y_1)^2}z^2 - \frac{d^2(x_0-x_1)}{(y_0-y_1)^2}z + \frac{d^4}{4(y_0-y_1)^2} -R^2 = 0
   \end{equation}
   Considering the quadratic formula for finding the solutions for this quadratic equation, its discriminant is as follows:
   $$\Delta = \Bigg(\frac{d^2(x_0-x_1)}{(y_0-y_1)^2}\Bigg)^2 - 4 \cdot \frac{d^2}{(y_0-y_1)^2} \cdot \Bigg[\frac{d^4}{4(y_0-y_1)^2} -R^2\Bigg] = $$
   $$= \frac{d^4}{(y_0-y_1)^4}[(x_0-x_1)^2 - d^2] + \frac{4 d^2 R^2}{(y_0-y_1)^2}$$
   Since $(x_0-x_1)^2 + (y_0-y_1)^2 = d^2$, we have that:
   $$\Delta = \frac{d^2}{(y_0-y_1)^2}(4R^2 - d^2)$$
   Clearly, a solution exists if and only if: $4R^2 - d^2 \geq 0 \Rightarrow -2R \leq d \leq 2R$. Since $d>0$, we have that: $\boxed{0 < d \leq 2R}$. Hence, the solution to the quadratic equation \ref{eq:quadratic-intersect} are as follows:
   $$z_{1,2} = \frac{\frac{d^2(x_0-x_1)}{(y_0-y_1)^2} \pm \frac{d}{y_0-y_1}\sqrt{4R^2 - d^2}}{2\frac{d^2}{(y_0-y_1)^2}} =$$
   $$= \frac{x_0-x_1}{2} \pm \frac{y_0-y_1}{2d}\sqrt{4R^2 - d^2}$$
   From symmetry, when denoting $z' := x_0 - x_4$ and substituting $z$ by $z'$ in Equation \ref{eq:quadratic-intersect}, it still holds. Thus, without loss of generality, we have that:
   $$\boxed{x_{3,4} = \frac{x_0+x_1}{2} \mp \frac{y_0-y_1}{2d}\sqrt{4R^2 - d^2}}$$
   Substituting this expression to Equation \ref{eq:y-3-4} yields the following:
   $$y_{3,4} = \frac{x_0-x_1}{y_0-y_1}\Bigg(\frac{x_0-x_1}{2} \pm \frac{y_0-y_1}{2d}\sqrt{4R^2 - d^2}\Bigg) + y_0 -\frac{d^2}{2(y_0-y_1)}=$$
   $$= \frac{(x_0-x_1)^2}{2(y_0-y_1)} \pm \frac{x_0-x_1}{2d}\sqrt{4R^2 - d^2} + y_0 -\frac{d^2}{2(y_0-y_1)}$$
   Noting that $(x_0-x_1)^2 + (y_0-y_1)^2 = d^2$, we have that:
   $$\boxed{y_{3,4} = \frac{y_0+y_1}{2} \pm \frac{x_0-x_1}{2d}\sqrt{4R^2 - d^2}}$$
   Denoting the equation for the desired locus by $y = bx+c$, the slope $b$ equals to:
   \begin{equation}
       \label{eq:slope}
       b = \frac{y_3-y_4}{x_3-x_4} = \frac{\frac{x_0-x_1}{d}\sqrt{4R^2 - d^2}}{-\frac{y_0-y_1}{d}\sqrt{4R^2 - d^2}}=\frac{x_1-x_0}{y_0-y_1}
   \end{equation}
   Given that the point $P_2$ lies on this line, we substitute its coordinates and the expression for $b$ to the locus' equation:
   $$\frac{y_0+y_1}{2} = \frac{x_1-x_0}{y_0-y_1} \cdot \frac{x_0+x_1}{2} + c \Rightarrow$$
   $$\Rightarrow c = \frac{(x_1+x_0)(x_1-x_0) - (y_0+y_1)(y_0-y_1)}{2(x_0-x_1)} =$$
   \begin{equation}
       \label{eq:locus-const}
       = \frac{x_1^2 - x_0^2 + y_0^2 - y_1^2}{2(x_0-x_1)}
   \end{equation}
   Thus, the desired locus is given by the following equation:
   \begin{equation}
       \label{eq:locus}
       \boxed{\mathcal{L} = \Big\{(x,y) \in \mathbb{R}^2 \Big| y = \frac{x_1-x_0}{y_0-y_1}x + \frac{x_1^2 - x_0^2 + y_0^2 - y_1^2}{2(x_0-x_1)} \wedge x \in [x_3,x_4]\Big\}}
   \end{equation}
\end{proof}

The \textit{Intersection Points Placement Method} is stated formally in Algorithm \ref{alg:intersection-method}. The algorithm receives two arguments - a connected component $C$ and an influencing agent $a_i$. Regarding the first line of the algorithm, if we were to assume by contradiction that for each pair $u,v \in C$ it holds that $||p_u(0)-p_v(0)|| > 2R$, then $a_u \notin N_v(t)$ is also satisfied - which is a contradiction to $C$ being a connected component. Therefore, such a pair always exists. According to Lemma \ref{lemma:discs-intersection}, the influencing agent is initially placed in the intersection of both flocking agents' neighborhoods. Thus, it will influence them both, and, consequently, the entire connected component. 

\begin{Algorithm}
    \label{alg:intersection-method}
    \captionof{algorithm}{\textbf{Intersection-Points-Placement-Method}($a_i$, $C$)}
    \centering
  \fbox{\begin{minipage}{\columnwidth}
	\begin{algorithmic}[1]
	    \State{Let $u,v \in C$ be a random pair, which satisfies $||p_u(0)-p_v(0)|| \leq 2R$.}
    	\State{The influencing agent $a_i$ is initially placed randomly along the linear line, connecting the intersection points of both flocking agents' neighborhoods.}
	\end{algorithmic}
  \end{minipage}}
\end{Algorithm}

In the worst case, in which the entire flock is strongly connected, the first line will run at most $\binom{k}{2}=\frac{k(k-1)}{2}$ times. The second line is performed with a time complexity of $O(1)$. Hence, the time complexity for Algorithm \ref{alg:intersection-method} is $O\big(\frac{k(k-1)}{2}\big)=O(k^2)$.

\section{Experiments}
\label{sec:experiment}
In this section we describe our experiments, testing the behaviour of the influencing agents, which should eventually lead to a spatial consensus in the observed flock. 
\subsection{Simulation Environment}
\label{sec:sim-env}
We situate our research within the Flockers domain of the MASON simulator \cite{luke2005mason}. This simulator encodes all the dynamics as they are described in the previous sections. Each agent points and moves in the direction of its current velocity vector. We made a few alterations to the the MASON Flockers domain, such that they will fit our needs. It was initially altered to also contain \textit{influencing agents}, which follow our predefined and fully controlled behaviours. Another modification was making the \textit{flocking agents} update their orientation according to either Equation \ref{eq:orientation-perron} or Equation \ref{eq:orientation-calcdiff}. For more realistic implications of the simulator, its toroidal feature was removed. That is, if an agent moves off of an edge of our domain, it will not reappear and will remain \textit{"lost"} forever.

\subsection{Placement Methods}
\label{sec:placment}
Genter \cite{genter2017fly} introduced a graph placement method, whose aim is to influence as many flocking agents as possible. They presented this method as an instance of the geometric set cover problem, thus making this process \textbf{NP-hard}. Accordingly, determining the "best" initial placements of the influencing agents is \textbf{NP-hard}. Due to its computational expense, we consider a \textit{random} placement of such agents at the area occupied by the flocking agents for the sake of running higher scale experiments.

For guaranteeing that the flocking neighbors graph does indeed consists of a single connected component, we consider two methods. For isolating the influence of a single influencing agent's initial position on a flock, we first consider a \textit{grid placement method} \cite{genter2015determining}. The flocking agents are initially placed at predefined, well-spaced, gridded positions. Since a grid of size $\ell \in \mathbb{N}$ can consist at most $\ell^2$ flocking agents, we utilize the smallest grid for which $k \leq \ell^2$ is satisfied. Each pair of adjacent flocking agents are initially placed to be within $R-1$ from each other.

The second one is a \textit{random placement method}, in which we initially randomly place a flocking agent $a_0$, where each of its coordinates are within a given interval. Afterwards, each successive flocking agent $a_i$ is randomly inserted to a position, which is within a radius of $R$ of the flocking agent $a_{i-1}$ ($1 \leq i \leq k-1$). Thus, the resulting flock constitute a single connected component. 

Throughout both algorithms, we calculate the maximal and minimal coordinates at which flocking agents are located with respect to both the $x$-axis and the $y$-axis, which we denote by $x_{min},x_{max},y_{min},y_{max}$. Those values are then extended by the visibility radius ($R$) and each influencing agent is initially placed randomly within the rectangular box formed by $x_{min}-R,x_{max}+R,y_{min}-R,y_{max}+R$ (which we denote by $AREA_{Flock}^+$). When initially placing a \textit{single} influencing agent randomly, we consider the rectangular box formed by $x_{min},x_{max},y_{min},y_{max}$ (which we denote by $AREA_{Flock}$), for the sake of increasing the probability at which it will indeed influence the flock.

\subsection{"Lost" Agents}
\label{sec:lost}
Genter \cite{genter2017fly} considered cases in which some flocking agents may become indefinitely separated from a flock with a switching topology. Hence, we formally define a \textit{"lost"} flocking agent as follows:
\begin{definition}
\label{def:lost-agent}
(\textit{"Lost" Flocking Agent})
A flocking agent $a_i$, aiming to converge to the orientation $\alpha$ (either willingly or not), is considered \textit{"lost"} if two criteria hold:
\begin{enumerate}
    \item There exists a subset of flocking agents with cardinality $0 < k' < k$ and all orientations are within $\epsilon$ of $\alpha$ for more than $T$ time steps.
    \item Given that $t^*$ is the time step at which the subset converged to $\alpha$: $|\theta_i(t^*) - \alpha|>\epsilon$. That is, the agent $a_i$ did not converge to $\alpha$ by the time the subset did converge.
\end{enumerate}
The entire flock is considered \textit{"lost"} if after $T_{Flock}$ time steps no flocking agents are within $\epsilon$ of $\alpha$. Consequently, the execution terminates, and is refered to as \textit{"totally lossy"}. Moreover, an execution containing a "lost" flocking agent is referred to as a \textit{"lossy"}.
\end{definition}

Regarding Genter \cite{genter2017fly}, we set $T = 200$ and $T_{Flock} = 2,800$.

\subsection{Experimental Setup}
\label{sec:exp-set}
The baseline experimental settings for variables are given in Table \ref{table:vars-def}. The variables' values are as they are presented in Table \ref{table:vars-def}, unless stated otherwise.

\begin{Table}
    \label{table:vars-def}
    \captionof{table}{Definition of Variables}
    \centering
  \fbox{\begin{minipage}{0.84\columnwidth}
      \begin{tabular}{|c|c|}
      \hline
        \textit{Variable} & \textit{Value} \\ \hline
        Domain Height & 300 \\\hline
        Domain Width & 300 \\\hline
        The visibility radius of each agent ($R$) & 10 \\\hline
        Velocity ($v_i$) & 0.2 \\\hline
      \end{tabular}
  \end{minipage}}
\end{Table}

Flocking agents are initially placed with random initial headings throughout the domain. In our experiments, we conclude that the flock has converged to an orientation $\alpha$ when every agent (that is not an influencing agent) is facing within 0.01 radians of $\alpha$. Other stopping criteria, such as when $90\%$ of the agents are facing within 0.01 radians of $\alpha$, could have also been used. Moreover, due to the involvement of randomness in our simulations, each point in all graphs corresponds to the average over $100$ consecutive executions.

Regarding the \textit{random placement method} of the flocking agents, there might be executions in which an influencing agent is initially placed in a manner that it does \textit{not} influence any flocking agent at all. In accordance, in cases where a single influencing agent is incorporated, convergence can \textbf{never} reached. Such executions are due to the randomness of the placement and we would like to only consider placing the influencing agents as we see fit. That is, we are motivated to initially place them such that they influence the flock, and therefore merely consider the executions in which convergence does occur. Thus, we utilize the \textit{Intersection Points placement method} proposed in Algorithm \ref{alg:intersection-method}, for the sake of ensuring that we indeed influence the connected component when considering the \textit{random placement method} for the flocking agents.

Regarding Section \ref{sec:influencing-connected-components}, influencing agents with a \textit{Face Desired Orientation behaviour} are employed. For each one of the presented graphs, we consider two executions, where each one corresponds to either the \textit{grid placement method} or the \textit{random placement method} described earlier.

\subsection{Experimental Results}
\label{sec:consensus-experimental}
In this subsection, we present the experimental results regarding consensus among a \textit{single} connected component, for both the fixed topology case and the switching topology case. Let us consider a flock willing to converge to a desired orientation $\alpha = \pi$.

\subsubsection{Fixed Topology}
\label{sec:fixed-experimental}
In this subsection, we assume that the resulting flocking neighbors graph is \textbf{fixed} and \textbf{strongly connected}. 

We first consider a \textit{single} influencing agent $a_k$ and the influence of its initial placement $p_k(0)$ on the number of time steps required until convergence. Figure \ref{fig:flocking} considers an increasing number of flocking agents and a \textit{single} influencing agent $a_k$, where all flocking agents are initially placed according to either the \textit{grid placement method} or the \textit{random placement method} described earlier. 

It was observed that, when the influencing agent was initially placed \textit{at the border} of $AREA_{Flock}$, the time steps required until the convergence was reached was \textit{higher} than the time steps required when it was initially placed \textit{at the interior} of $AREA_{Flock}$. For instance, regarding the \textit{grid placement method}, when utilizing $50$ flocking agents, we observed that there was an execution, for which it took only $7535$ time steps until convergence, when $p_k(0)$ was the flock's \textit{interior}, but that there was an execution, for which it took $34,838$ time steps, when $p_k(0)$ was at the flock's \textit{border}. Theoretically, a possible explanation could be that the number of \textbf{directly} influenced flocking agents stems from the influencing agent's initial placement.

Furthermore, regarding the \textit{random placement method}, we noted that the number of time steps was influenced by an additional factor - \textbf{the density} at which the flocking agents were initially placed, rather than just being influenced by $p_k(0)$. Given that the flocking agents were placed in a \textit{high} density, we observed that the convergence was indeed \textit{faster} than the one reached in a \textit{lower} density. For instance, when utilizing $50$ flocking agents, we observed that there was an execution, for which it took only $738$ time steps until convergence when their density was quite \textit{high}, but there was an execution, for which it took $63,560$ time steps when the density was quite \textit{low}. Theoretically, the \textit{higher} the density at which they are initially placed, the \textit{higher} the number of flocking agents lying within each flocking agent's neighborhood, leading to a greater influence upon it.

We now consider $k=100$ flocking agents and investigate the influence of the number of influencing agents $m$ upon the number of time steps required until convergence. Figure \ref{fig:influencing} considers an increasing number of influencing agents and $k=100$ flocking agents, where all flocking agents are initially placed according to either the \textit{grid placement method} or the \textit{random placement method} described earlier. 

We noted that the number of time steps was influenced by their \textbf{coverage} of the area occupied by the flocking agents. For instance, regarding the \textit{random placement method}, when utilizing $m=10$ influencing agents which are evenly spread across $AREA_{Flock}$, we observed that there was an execution, for which it took $606$ time steps until convergence, but there was an execution, for which took $29,581$ time steps when they were quite crowded around a specific set of flocking agents. Theoretically, as long as they are evenly spread across this area, they influence \textit{more} flocking agents directly, resulting in \textit{less} time steps required until convergence.

Furthermore, it can be observed in Figure \ref{fig:influencing} that when the number of influencing agents $m$ reaches closer to the number of flocking agents $k=100$, the number of time steps required for convergence did not change drastically. For instance, regarding the \textit{grid placement method}, when utilizing $80,90$ influencing agents, we observed that it took $612.7,538.95$ time steps on average until convergence (respectively).

Regarding the \textit{random placement method}, we noted that there existed executions, in which several influencing agents weren't initially placed at the neighborhood of any flocking agent. Thus, no flocking agent was influenced, either directly or indirectly, by such influencing agents. Consequently, the number of time steps required for convergence was \textit{higher} in such cases. Theoretically, due to the fact that the area occupied by the flocking agents when regarding the random placement method is \textit{smaller}, they are scattered across the rectangular box formed by $x_{min}-R,x_{max}+R,y_{min}-R,y_{max}+R$ at a \textit{higher} density. Hence, the flocking agents' coverage of the mentioned box becomes \textit{lower}, and therefore there are more positions, at which placing an influencing agent won't influence any of the flocking agents.

\subsubsection{Switching Topology}
\label{sec:switching-experimental}

In this subsection, we assume that the resulting flocking neighbors graph is \textbf{switching} and that $\mathcal{G}(0)$ is \textbf{strongly connected}.

We first consider a \textit{single} influencing agent $a_k$ and the influence of its initial placement $p_k(0)$ on the average number of \textit{"lost"} flocking agents (Figure \ref{fig:flocking-lost}) and the number of execution in which \textit{all} flocking agents are \textit{"lost"} (Figure \ref{fig:flocking-lost-executions}). All flocking agents are initially placed according to either the \textit{grid placement method} or the \textit{random placement method} described earlier. Accordingly, we infer that the \textit{grid placement method} performs \textit{better} than the \textit{random placement method} (on average).
We now consider $k=50$ flocking agents and investigate the influence of the number of influencing agents $m$ upon the average number of \textit{"lost"} flocking agents (Figure \ref{fig:influencing-lost}) and the number of time steps until all the other agents converged (Figure \ref{fig:influencing-lost-steps}). All flocking agents are initially placed according to either the \textit{grid placement method} or the \textit{random placement method} described earlier. 

It can be observed by Figure \ref{fig:switching-influencing} that a result, which is \textit{opposite} to the one inferred from Figure \ref{fig:switching}, is achieved. Indeed, when considering the \textit{grid placement method}, there were \textbf{0 "totally lossy" executions} for each value of $m$, although they were \textit{all} "lossy" ones. In contrast, when considering the \textit{random placement method}, there were \textbf{8 "totally lossy" executions} for $m=10$, but for $m=60,70,80,90$ there were \textbf{5, 2, 9, 4 executions in which \textit{all} flocking agents converged to the desired orientation} (respectively). This is due to both the \textbf{density} and the \textbf{coverage} factors mentioned in Subsection \ref{sec:fixed-experimental}. That is, the influencing agents' area of influence becomes \textit{smaller} when the \textit{random placement method} is utilized. Thus, the density at which they are initially placed  and their coverage of the flock are \textit{higher} as $m$'s value rises, therefore making the \textit{random placement method}'s performance more superior.
\end{multicols}

\begin{figure}[h!]
    \begin{tabular}{cc}
        \begin{subfigure}[b]{0.5\textwidth}
            \centering
            \begin{tikzpicture}[yscale=0.54,xscale=0.54]
                \begin{axis}[
                    xlabel={Number of Flocking Agents ($k$)},
                    ylabel={Number of Time Steps},
                    xmin=0, xmax=60,
                    ymin=0, ymax=20000,
                    xtick={0,20,40,60},
                    ytick={0,2000,4000,6000,8000,10000,12000,14000,16000,18000,20000},
                    legend pos=north west,
                    ymajorgrids=true,
                    grid style=dashed,
                ]
                
                \addplot[
                    color=blue,
                    mark=square,
                    ]
                    coordinates {
                    (10,638.88)(20,2078.14)(30,5364.04)(40,9327.34)(50,18365.39)
                    };
                
                \addplot[
                    color=red,
                    mark=square,
                    ]
                    coordinates {
                    (10,457.75)(20,1773.18)(30,3458.83)(40,6996.94)(50,11815.15)
                    };
                
                \legend{Grid, Random}
                
                \end{axis}
                \end{tikzpicture}
                \caption{A Single influencing Agent's Initial Placement}
                \label{fig:flocking}
        \end{subfigure}
        
        &
        
        \begin{subfigure}[b]{0.5\textwidth}
            \centering
                \begin{tikzpicture}[yscale=0.54,xscale=0.54]
                \begin{axis}[
                    xlabel={Number of influencing Agents ($m$)},
                    ylabel={Number of Time Steps},
                    xmin=0, xmax=100,
                    ymin=0, ymax=7000,
                    xtick={0,20,40,60,80,100},
                    ytick={0,1000,2000,3000,4000,5000,6000,7000},
                    legend pos=north east,
                    ymajorgrids=true,
                    grid style=dashed,
                ]
                
                \addplot[
                    color=blue,
                    mark=square,
                    ]
                    coordinates {
                    (10,5571.17)(20,2715.41)(30,1767.76)(40,1358.31)(50,1044.05)(60,839.02)(70,727.12)(80,612.7)(90,538.95)
                    };
                
                \addplot[
                    color=red,
                    mark=square,
                    ]
                    coordinates {
                    (10,6517.03)(20,3796.77)(30,1976.81)(40,1289.08)(50,843.96)(60,849.62)(70,538.57)(80,489.1)(90,538.92)
                    };
                
                \legend{Grid, Random}
                
                \end{axis}
                \end{tikzpicture}
                \caption{Increasing Number of influencing Agents - $k=100$}
                \label{fig:influencing}
        \end{subfigure}
        
    \end{tabular}
    \caption{Figure \ref{fig:flocking} shows a comparison of the experimental results when considering an increasing number of flocking agents and a \textit{single} influencing agent. Figure \ref{fig:influencing} shows a comparison of the experimental results when considering an increasing number of influencing agents and $k=100$ flocking agents.}
    \label{fig:fixed}
\end{figure}
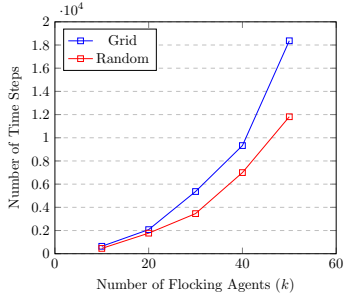
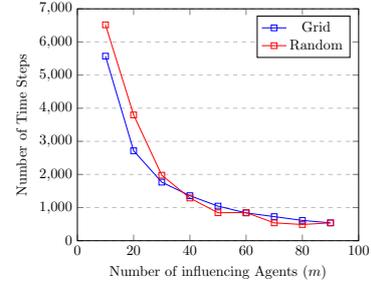 

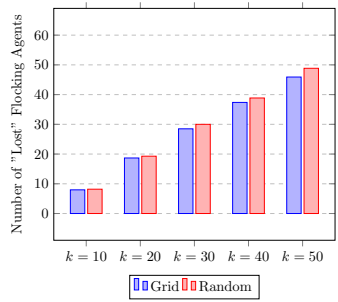
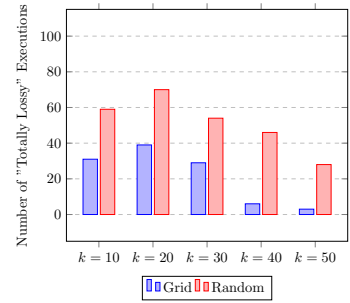
\begin{figure}[h!]
    \begin{tabular}{cc}
        \begin{subfigure}[b]{0.5\textwidth}
            \centering
                \begin{tikzpicture}[yscale=0.54,xscale=0.54]
                \begin{axis}[
                    ybar ,
                        ymin=0, 
                        ymax=60,
                        ytick={0,10,20,30,40,50,60},
                        enlargelimits=0.15,
                        legend style={at={(0.5,-0.15)},
                        anchor=north,legend columns=-1},
                        ylabel={Number of "Lost" Flocking Agents},
                        symbolic x coords={$k=10$,$k=20$,$k=30$,$k=40$,$k=50$},
                        xtick=data,
                        nodes near coords align={vertical},
                        ymajorgrids=true,
                        grid style=dashed,
                        ]
                    \addplot coordinates {($k=10$,7.96) ($k=20$,18.68) ($k=30$,28.49) ($k=40$,37.37) ($k=50$,45.91)};
                    \addplot coordinates {($k=10$,8.18) ($k=20$,19.28) ($k=30$,30) ($k=40$,38.86) ($k=50$,48.85)};
                    \legend{Grid, Random}
                \end{axis}
                \end{tikzpicture}
            \caption{Average \# "Lost" Agents}
            \label{fig:flocking-lost}
        \end{subfigure}
        & 
        \begin{subfigure}[b]{0.5\textwidth}
            \centering
                \begin{tikzpicture}[yscale=0.54,xscale=0.54]
                \begin{axis}[
                        ybar ,
                        ymin=0, 
                        ymax=100,
                        ytick={0,20,40,60,80,100},
                        enlargelimits=0.15,
                        legend style={at={(0.5,-0.15)},
                        anchor=north,legend columns=-1},
                        ylabel={Number of "Totally Lossy" Executions},
                        symbolic x coords={$k=10$,$k=20$,$k=30$,$k=40$,$k=50$},
                        xtick=data,
                        nodes near coords align={vertical},
                        ymajorgrids=true,
                        grid style=dashed,
                        ]
                    \addplot coordinates {($k=10$,31) ($k=20$,39) ($k=30$,29) ($k=40$,6) ($k=50$,3)};
                    \addplot coordinates {($k=10$,59) ($k=20$,70) ($k=30$,54) ($k=40$,46) ($k=50$,28)};
                    \legend{Grid, Random}
                \end{axis}
                
                \end{tikzpicture}
            \caption{\# "Totally Lossy" Executions}
            \label{fig:flocking-lost-executions}
        \end{subfigure}
        
    \end{tabular}
    \caption{Those graphs consider a \textit{single} influencing agent $a_k$ and the influence of its initial placement $p_k(0)$ on the average number of \textit{"lost"} flocking agents (Figure \ref{fig:flocking-lost}) and the number of execution in which \textit{all} flocking agents are lost (Figure \ref{fig:flocking-lost-executions}).}
    \label{fig:switching}
\end{figure} 

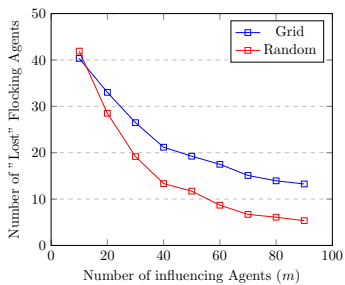
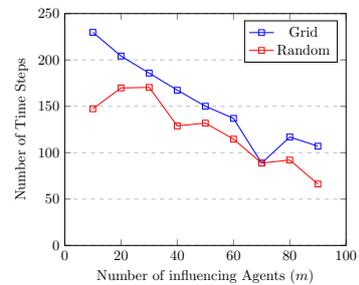
\begin{figure}[h!]
    \begin{tabular}{cc}
        \begin{subfigure}[b]{0.5\textwidth}
            \centering
            \begin{tikzpicture}[yscale=0.54,xscale=0.54]
                \begin{axis}[
                    xlabel={Number of influencing Agents ($m$)},
                    ylabel={Number of "Lost" Flocking Agents},
                    xmin=0, xmax=100,
                    ymin=0, ymax=50,
                    xtick={0,20,40,60,80,100},
                    ytick={0,10,30,20,40,50},
                    legend pos=north east,
                    ymajorgrids=true,
                    grid style=dashed,
                ]
                
                \addplot[
                    color=blue,
                    mark=square,
                    ]
                    coordinates {                    (10,40.34)(20,33.02)(30,26.48)(40,21.18)(50,19.25)(60,17.5)(70,15.1)(80,13.94)(90,13.26)
                    };
                
                \addplot[
                    color=red,
                    mark=square,
                    ]
                    coordinates {                    (10,41.91)(20,28.5)(30,19.17)(40,13.35)(50,11.71)(60,8.68)(70,6.7)(80,6.09)(90,5.35)
                    };
                
                \legend{Grid, Random}
                
                \end{axis}
                \end{tikzpicture}
            \caption{Average \# "Lost" Agents}
            \label{fig:influencing-lost}
        \end{subfigure}
        & 
        \begin{subfigure}[b]{0.5\textwidth}
            \centering
                \begin{tikzpicture}[yscale=0.54,xscale=0.54]
                \begin{axis}[
                    xlabel={Number of influencing Agents ($m$)},
                    ylabel={Number of Time Steps},
                    xmin=0, xmax=100,
                    ymin=0, ymax=250,
                    xtick={0,20,40,60,80,100},
                    ytick={0,50,100,150,200,250},
                    legend pos=north east,
                    ymajorgrids=true,
                    grid style=dashed,
                ]
                
                \addplot[
                    color=blue,
                    mark=square,
                    ]
                    coordinates {                    (10,229.77)(20,204.18)(30,185.85)(40,167.39)(50,150.14)(60,136.99)(70,88.82)(80,116.89)(90,107.05)
                    };
                
                \addplot[
                    color=red,
                    mark=square,
                    ]
                    coordinates {                    (10,147.31)(20,169.84)(30,170.4)(40,128.88)(50,131.91)(60,114.65)(70,88.82)(80,92.07)(90,66.2)
                    };
                
                \legend{Grid, Random}
                
                \end{axis}
                \end{tikzpicture}
            \caption{\# Time Steps}
            \label{fig:influencing-lost-steps}
        \end{subfigure}
        
    \end{tabular}
    \caption{Those graphs consider $k=50$ flocking agents and investigate the influence of the number of influencing agents $m$ on the average number of \textit{"lost"} flocking agents (Figure \ref{fig:influencing-lost}) and the number of time steps until all the other agents converged (Figure \ref{fig:influencing-lost-steps}).}
    \label{fig:switching-influencing}
\end{figure} 

\begin{multicols}{2}
\section{Conclusions}
\label{sec:conclusions-future}
We provided the convergence analysis of a consensus protocol for a network of integrators with directed information flow and either a fixed or a switching topology. Our analysis relies on several tools from algebraic graph theory, matrix theory, control theory and Lyapunov stability. We examined the problem of forcing the flock to reach a {\em desired} orientation by inserting one or \textit{more} agents, referred to as {\em influencing agents}, into the flock. We proved that influencing agents with a \textit{Face Desired Orientation behaviour} are sufficient for guaranteeing consensus, while employing them into a swarm of flocking agents that follow the Vicsek Model \cite{vicsek1995novel}. 

We have also presented that continuous-time update rules converge exponentially \textit{faster} than discrete-time ones, thus yielding the benefits of the continuous-time case. Furthermore, for a switching topology with a continuous-time update rule, we proposed a more concrete convergence rate than the one introduced by Olfati-Saber and Murray \cite{olfati2004consensus}. 
Moreover, we introduced the \textit{Intersection Points Placement Method}, which guarantees in polynomial time that a \textit{single} influencing agent can be inserted into the flock in a manner that it will necessarily influence an entire single connected component.



\bibliographystyle{plain}
\bibliography{sample}

\end{multicols}
\end{document}